\documentclass[11pt, draftcls, onecolumn]{IEEEtran}%
\usepackage{amssymb}
\usepackage{amsmath}
\usepackage{graphicx}
\usepackage{amsfonts}
\usepackage{psfrag}
\usepackage{dsfont}%
\newtheorem{thm}{Theorem}[section]

\newtheorem{prop}[thm]{Proposition}
\newtheorem{defn}{Definition}[section]
\newtheorem{rem}{Remark}[section]
\begin{document}

\title{Multiple Multicasts with the Help of a Relay}
\author{Deniz G\"{u}nd\"{u}z, Osvaldo Simeone, Andrea J. Goldsmith, \\ H. Vincent Poor and Shlomo Shamai (Shitz) \thanks{Deniz G\"{u}nd\"{u}z is with the Dept. of Electrical Engineering, Stanford University, Stanford, CA 94305, USA and the Dept. of Electrical Engineering, Princeton University, Princeton, NJ 08544 USA (email: dgunduz@princeton.edu).} \thanks{Osvaldo Simeone is with the CWCSPR, New Jersey Institute of Technology, Newark, NJ 07102 USA (email: osvaldo.simeone@njit.edu).} \thanks{Andrea J. Goldsmith is with the Dept. of Electrical Engineering, Stanford University, Stanford, CA 94305, USA.}\thanks{H. Vincent Poor is with the Dept. of Electrical Engineering, Princeton University, Princeton, NJ 08544 USA.}\thanks{Shlomo Shamai (Shitz) is with the Dept. of Electrical Engineering, Technion, Haifa, 32000, Israel.}\thanks{This work was supported
by NSF under grants CNS-06-26611 and CNS-06-25637, the DARPA ITMANET program
under grant 1105741-1-TFIND and the ARO under MURI award W911NF-05-1-0246. The
work of S. Shamai has been supported by the Israel Science Foundation and the
European Commission in the framework of the FP7 Network of Excellence in
Wireless COMmunications NEWCOM++.}}

\maketitle
\vspace{-.5in}
\begin{abstract}
The problem of simultaneous multicasting of multiple messages with the help of a relay
terminal is considered. In particular, a model is studied in which a relay
station simultaneously assists two transmitters in multicasting their
independent messages to two receivers. The relay may also have an independent
message of its own to multicast. As a first step to address this general model, referred
to as the compound multiple access channel with a relay (cMACr), the capacity
region of the multiple access channel with a ``cognitive'' relay is
characterized, including the cases of partial and rate-limited cognition.
Then, achievable rate regions for the cMACr model are presented based on
decode-and-forward (DF) and compress-and-forward (CF) relaying strategies.
Moreover, an outer bound is derived for the special case, called the cMACr without cross-reception, in which each
transmitter has a direct link to one of the receivers while the connection to
the other receiver is enabled only through the relay terminal. The capacity region is characterized for a binary modulo additive
cMACr without cross-reception, showing the optimality of binary linear block codes, thus highlighting the
benefits of physical layer network coding and structured codes. Results are
extended to the Gaussian channel model as well, providing achievable
rate regions for DF and CF, as well as for a structured code design based on
lattice codes. It is shown that the performance with lattice codes approaches
the upper bound for increasing power, surpassing the rates achieved by the
considered random coding-based techniques.

\end{abstract}

\thispagestyle{empty}

\section{Introduction}

\label{s:intro}

Consider two non-cooperating satellites each multicasting radio/TV signals to
users on Earth. The coverage area and the quality of the transmission is
generally limited by the strength of the direct links from the satellites to
the users. To extend coverage, to increase capacity or to improve robustness,
a standard solution is that of introducing relay terminals, which may be other
satellite stations or stronger ground stations (see Fig. \ref{f:model_img}).
The role of the relay terminals is especially critical in scenarios in which
some users lack a direct link from any of the satellites. Moreover, it is
noted that the relays might have their own multicast traffic to transmit. A
similar model applies in the case of non-cooperating base stations
multicasting to mobile users in different cells: here, relay terminals located
on the cell boundaries may help each base station reach users in the
neighboring cells.

Cooperative transmission (relaying) has been extensively studied in the case
of two transmitting users, both for a single user with a dedicated relay
terminal \cite{Cover:IT:79}, \cite{Kramer:IT:05} and for two cooperating users
\cite{Willems:IT:83}. Extensions to scenarios with multiple users are
currently under investigation \cite{Kramer:IT:05}, \cite{Liang:IT:07} - \cite{Maric:MILCOM:07}. In this
work, we aim at studying the impact of cooperation in the setup of Fig.
\ref{f:model_img} that consists of two source terminals simultaneously
multicasting independent information to two receivers in the presence of a
relay station. While the source terminals cannot directly cooperate with each
other, the relay terminal is able to support both transmissions simultaneously
to enlarge the multicast capacity region of the two transmitters. Moreover, it
is assumed that the relay station is also interested in multicasting a local
message to the two receivers (see Fig. \ref{model}).

\begin{figure}[ptb]
\begin{center}
\includegraphics[width=5in]{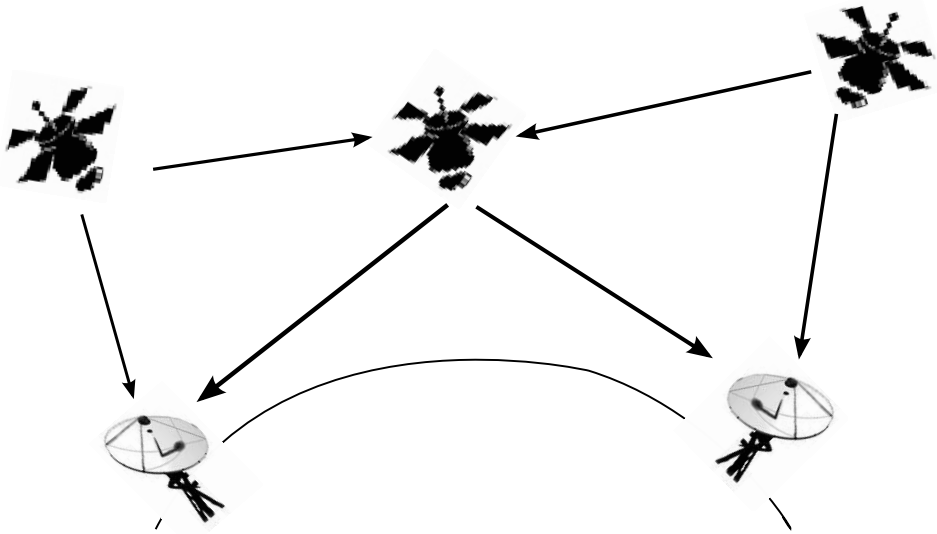}
\end{center}
\caption{Illustration for an application of the compound multiple access
channel with a relay.}%
\label{f:model_img}%
\end{figure}

The model under study is a \emph{compound multiple access channel with a
relay} (cMACr) and can be seen as an extension of several fundamental channel
models, such as the multiple access channel (MAC), the broadcast channel (BC)
and the relay channel (RC). The main goal of this work is to adapt basic
transmission strategies known from these key scenarios to the channel at hand
and to identify special cases of the more general model for which conclusive
capacity results can be obtained.

Below, we summarize our contributions:

\begin{itemize}
\item We start our analysis by studying a simplified version of the cMACr that
consists of a MAC with a ``cognitive'' relay (see Fig. \ref{model2}). In this
scenario the cognitive relay is assumed to be aware of both transmitters'
messages non-causally. We provide the capacity region for this model and
several extensions. While interesting on its own, this setup enables us to
conveniently introduce the necessary tools to address the analysis of the
cMACr. As an intermediate step between the cognitive relay model and the more
general model of cMACr, we also consider the relay with finite capacity
unidirectional links from the transmitters and provide the corresponding
capacity region.

\item We provide achievable rate regions for the cMACr model with decode-and-forward (DF) and
compress-and-forward (CF) relaying. In the CF scheme, the relay, instead of
decoding the messages, quantizes and broadcasts its received signal. This
corresponds to the joint source-channel coding problem of broadcasting a
common source to two receivers, each with its own correlated side information,
in a lossy fashion, studied in \cite{Nayak:IT:08}. This result indicates that
the pure channel coding rate regions for certain multi-user networks can be
improved by exploiting related joint source-channel coding techniques.

\item The similarity between the underlying scenario and the classical
butterfly example in network coding \cite{Ahlswede:IT:00} is evident, despite the fact that we have
multiple sources and a more complicated network with broadcasting constraints
and multiple access interference. Yet, we can still benefit from physical
layer coding techniques that exploit the network coding techniques. In order to
highlight the possibility of physical layer network coding, we focus on a special cMACr in which each source's signal is received directly by only one of the destinations, while the other destination is reached through the relay. This special model is called the \emph{cMACr without cross-reception}. We provide an outer bound for this
setting and show that it matches the DF achievable region, apart from an
additional sum rate constraint at the relay terminal. This
indicates the suboptimality of enforcing the relay to decode both messages,
and motivates a coding scheme that exploits the network coding aspects in the
physical layer.

\item Based on the observation above, we are interested in leveraging the network structure by exploiting ``structured codes''. We then focus on a modulo additive binary version of the cMACr, and characterize its capacity region, showing that it is achieved by binary linear block codes. In this scheme, the relay only decodes the binary sum of the transmitters' messages, rather than decoding each individual message. Since the receiver $1$ ($2$) can decode the message of transmitter $1$ ($2$) directly without the help of the relay, it is sufficient for the relay to forward only the binary sum. Similar to \cite{Korner:IT:79}, \cite{Nazer:ETT:08}, \cite{Knopp:GDR:07}, this result highlights the importance of structured codes in achieving the capacity region of certain multi-user networks.

\item Finally, we extend our results to the Gaussian case, and present a
comparison of the achievable rates and the outer bound. Additionally, we
extend the structured code approach to the Gaussian channel setting by proposing an achievable scheme based on nested lattice codes. We show
that, in the case of symmetric rates from the transmitters, nested lattice
coding improves the achievable rate significantly compared to the considered
random coding schemes in the moderate to high power regime.
\end{itemize}

The cMACr of Fig. \ref{model} can also been seen as a generalization of a
number of other specific channels that have been studied extensively in the literature. To start with, if there is no relay
terminal available, our model reduces to the compound multiple access channel
whose capacity is characterized in \cite{Ahlswede:1974}. Moreover, if there is
only one source terminal, it reduces to the dedicated relay broadcast
channel with a single common message explored in \cite{Kramer:IT:05},
\cite{Liang:IT:07}: Since the capacity is not known even for the simpler case
of a relay channel \cite{Cover:IT:79}, the capacity for the dedicated relay
broadcast channel remains open as well. If we have two sources but a single
destination, the model reduces to the multiple access relay channel model
studied in \cite{Kramer:IT:05}, \cite{Kramer:ISIT:00} whose
capacity region is not known in the general case either. Furthermore, if we
assume that transmitter 1 (and 2) has an orthogonal side channel of infinite
capacity to receiver 1 (2), then we can equivalently consider the message of
transmitter 1 (2) to be known in advance at receiver 1 (2) and the
corresponding channel model becomes equivalent to the restricted two-way relay
channel studied in \cite{Rankov:Asil:05}, \cite{Gunduz:All:08},
\cite{Wilson:IT:08}, and \cite{Nam:IZS:08}.

The cMACr model is also studied in \cite{Maric:MILCOM:07}, where
DF and amplify-and-forward (AF) based protocols are
analyzed. Another related problem is the interference relay channel model
studied in \cite{Sahin:Asil:07}, \cite{Maric:ISIT:08},
\cite{Sridharan:ISIT:08}: Note that, even though the interference channel
setup is not obtained as a special case of our model, achievable rate regions
proposed here can serve as inner bounds for that setup as well.

\textit{Notation}: To simplify notation, we will sometimes use the shortcut:
$x_{\{S\}}=(x_{i})_{i\in S}.$ We employ standard conventions (see, e.g.,
\cite{Cover:IT:79}), where the probability distributions are defined by the
arguments, upper-case letters represent random variables and the corresponding
lower-case letters represent realizations of the random variables. We will
follow the convention of dropping subscripts of probability distributions if
the arguments of the distributions are lower case versions of the
corresponding random variables. The superscripts identify the number of
samples to be included in a given vector, e.g., $y_{1}^{j-1}=[y_{1,1}\cdots
y_{1,j-1}].$

\psfrag{W1}{$W_1$}\psfrag{W2}{$W_2$} \psfrag{W3}{$W_3$}
\psfrag{X1}{$X_1$}\psfrag{X2}{$X_2$}%
\psfrag{X3}{$X_3$} \psfrag{Y1}{$Y_1$}\psfrag{Y2}{$Y_2$}\psfrag{Y3}{$Y_3$} %
\psfrag{pxy}{${\textstyle p(y_1,y_2,y_3|x_1,x_2,x_3)}$}
\psfrag{hW1}{${\scriptstyle \hat{W}_1(1), \hat{W}_2(1), \hat{W}_3(1)}$}
\psfrag{hW2}{${\scriptstyle \hat{W}_1(2), \hat{W}_2(2), \hat{W}_3(2)}$} \psfrag{Rel}{\small Relay}%
\psfrag{S1}{\small Source 1}\psfrag{S2}{\small Source 2} \psfrag{D1}{\small Dest. 1}%
\psfrag{D2}{\small Dest. 2} \psfrag{Ch}{cMACr}
\begin{figure}[tbp]
\centering \includegraphics[width=5in]{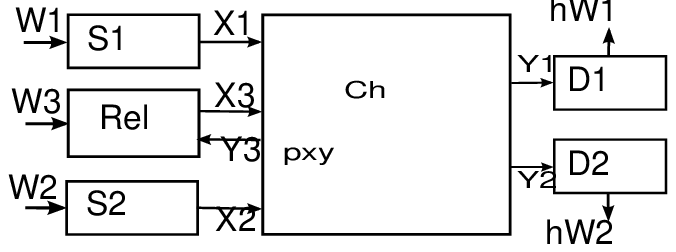}
\caption{A compound MAC with a relay (cMACr).}
\label{model}
\end{figure}

The rest of the paper is organized as follows. The system model is introduced in
Section \ref{s:system_model}. In Section \ref{s:MAC_cogrel} we study the
multiple access channel with a cognitive relay, and provide the capacity
region for this model and several extensions. The compound multiple access
channel with a relay is studied in Section \ref{s:cMACr}, in which inner and
outer bounds are provided using decode-and-forward and compress-and-forward
type relaying strategies. Section \ref{sec:linear binary} is devoted to a special
binary additive cMACr model. For this model, we characterize the capacity
region and show that the linear binary block codes can achieve any point in
the capacity region, while random coding based achievability schemes have
suboptimal performance. In Section \ref{s:Gaussian}, we analyze Gaussian
channel models for both the MAC with a relay setup and the general cMACr
setup. We apply lattice coding/decoding for the cMACr and show that it
improves the achievable symmetric rate value significantly, especially for the
high power regime. Section \ref{s:conc} concludes the paper followed by the
appendices where we have included the details of the proofs.

\section{System Model}

\label{s:system_model}

A compound multiple access channel with relay consists of three
channel input alphabets $\mathcal{X}_{1}$, $\mathcal{X}_{2}$ and
$\mathcal{X}_{3}$ of transmitter 1, transmitter 2 and the relay, respectively,
and three channel output alphabets $\mathcal{Y}_{1}$, $\mathcal{Y}_{2}$ and
$\mathcal{Y}_{3}$ of receiver 1, receiver 2 and the relay, respectively. We
consider a discrete memoryless time-invariant channel without feedback, which
is characterized by the transition probability $p(y_{1},y_{2},y_{3}%
|x_{1},x_{2},x_{3})$ (see Fig. \ref{model}). Transmitter $i$ has message
$W_{i}\in\mathcal{W}_{i}$, $i=1,2$, while the relay terminal also has a
message $W_{3}\in\mathcal{W}_{3}$ of its own, all of which need to be
transmitted reliably to both receivers. Extension to a Gaussian model will be
considered in Sec. \ref{s:Gaussian}.

\begin{defn}
A $(2^{nR_{1}},2^{nR_{2}},2^{nR_{3}},n)$ code for the cMACr consists of three
sets $\mathcal{W}_{i}=\{1,\ldots,2^{nR_{i}}\}$ for $i=1,2,3$, two encoding
functions $f_{i}$ at the transmitters, $i=1,2$,
\begin{equation}
f_{i}:\mathcal{W}_{i}\rightarrow\mathcal{X}_{i}^{n},
\end{equation}
a set of (causal) encoding functions $g_{j}$ at the relay, $j=1,\ldots,n$,
\begin{equation}
g_{j}:\mathcal{W}_{3}\times\mathcal{Y}_{3}^{j-1}\rightarrow\mathcal{X}_{3},
\end{equation}
and two decoding functions $h_{i}$ at the receivers, $i=1,2$,
\begin{equation}
h_{i}:\mathcal{Y}_{i}^{n}\rightarrow\mathcal{W}_{1}\times\mathcal{W}_{2}%
\times\mathcal{W}_{3}.
\end{equation}

\end{defn}


We assume that the relay terminal is capable of full-duplex operation, i.e.,
it can receive and transmit at the same time instant. The joint distribution of the
random variables factors as
\begin{equation}
p(w_{\{1,2,3\}},x_{\{1,2,3\}}^{n},y_{\{1,2,3\}}^{n})=\prod_{i=1}^{3}%
p(w_{i})\cdot p(x_{1}^{n}|w_{1})p(x_{2}^{n}|w_{2})\prod_{j=1}^{n}%
p(x_{3j}|y_{3}^{j-1},w_{3})p(y_{\{1,2,3\}j}|x_{\{1,2,3\}j}%
).\label{distribution 1}%
\end{equation}

The average probability of block error for this code is defined as
\[
P_{e}^{n}\triangleq\frac1{2^{n(R_{1}+R_{2}+R_{3})}} \sum_{(W_{1}, W_{2},
W_{3}) \in\mathcal{W}_{1} \times\mathcal{W}_{2} \times\mathcal{W}_{3}}
\Pr\left[  \bigcup_{i=1,2}\{(\hat{W}_{1}(i),\hat{W}_{2}(i),\hat{W}_{3}%
(i))\neq(W_{1},W_{2},W_{3})\}\right]  .
\]

\begin{defn}
A rate triplet $(R_{1},R_{2},R_{3})$ is said to be \emph{achievable} for the
cMACr if there exists a sequence of $(2^{nR_{1}},2^{nR_{2}},2^{nR_{3}},n)$
codes with $P_{e}^{n}\rightarrow0$ as $n\rightarrow\infty$.
\end{defn}

\begin{defn}
The \emph{capacity region} $\mathcal{C}$ for the cMACr is the closure of the
set of all achievable rate triplets.
\end{defn}

\section{MAC with a Cognitive Relay}\label{s:MAC_cogrel}

Before addressing the more general cMACr model, in this section we study the
simpler MAC with a cognitive relay scenario shown in Fig. \ref{model2}. This
model, beside being relevant on its own, enables the introduction of tools and
techniques of interest for the cMACr. The model differs from the cMACr in that
the messages $W_{1}$ and $W_{2}$ of the two users are assumed to be
non-causally available at the relay terminal (in a ``cognitive'' fashion
\cite{Devroye:IT:06}) and there is only one receiver
($\mathcal{Y}_{2}=\mathcal{Y}_{3}=\emptyset$ and $\mathcal{Y}=\mathcal{Y}_{1}%
$). Hence, the encoding function at the relay is now defined as $f_{3}%
:\mathcal{W}_{1}\times\mathcal{W}_{2}\times\mathcal{W}_{3}\rightarrow
\mathcal{X}_{3}^{n}$, the discrete memoryless channel is characterized by the
conditional distribution $p(y|x_{1},x_{2},x_{3})$ and the average block error
probability is defined accordingly for a single receiver. Several extensions
of the basic model of Fig. \ref{model2} will also be considered in this
section. The next proposition provides the capacity region for the MAC with a
cognitive relay.

\psfrag{W1}{$W_1$}\psfrag{W2}{$W_2$}\psfrag{W3}{$W_3$}
\psfrag{X1}{$X_1$}\psfrag{X2}{$X_2$}\psfrag{X3}{$X_3$}
\psfrag{Y1}{$Y$}
\psfrag{pxy}{${\textstyle p(y|x_1,x_2,x_3)}$}
\psfrag{hW1}{${\hat{W}_1, \hat{W}_2, \hat{W}_3}$}
\psfrag{Cog}{\small Cognitive}\psfrag{Rel}{\small Relay}%
\psfrag{S1}{\small Source 1}\psfrag{S2}{\small Source 2} \psfrag{D1}{\small Destination}
\psfrag{Ch}{Channel}
\begin{figure}[ptb]
\begin{center}
\includegraphics[width=5in]
{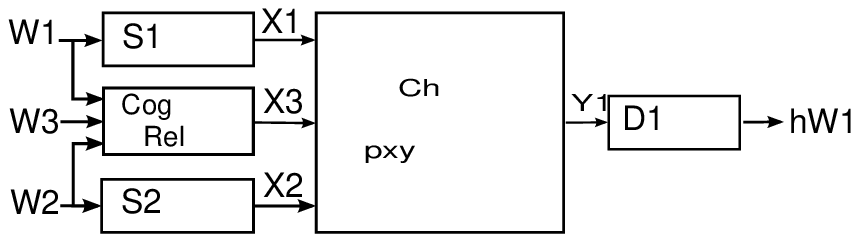}
\end{center}
\caption{MAC with a cognitive relay.}%
\label{model2}%
\end{figure}

\begin{prop}
\label{prop:1} For the MAC with a cognitive relay, the capacity region is the
closure of the set of all non-negative $(R_{1},R_{2},R_{3})$ satisfying
\begin{subequations}
\label{region cognitive}%
\begin{align}
R_{3} &  \leq I(X_{3};Y|X_{1},X_{2},U_{1},U_{2},Q), \\
R_{1}+R_{3} &  \leq I(X_{1},X_{3};Y|X_{2},U_{2},Q), \\
R_{2}+R_{3} &  \leq I(X_{2},X_{3};Y|X_{1},U_{1},Q),
\end{align}
and
\begin{align}
R_{1}+R_{2}+R_{3} &  \leq I(X_{1},X_{2},X_{3};Y|Q)  \label{non-conditional}%
\end{align}
for some joint distribution of the form
\end{subequations}
\begin{equation}
p(q)p(x_{1},u_{1}|q)p(x_{2},u_{2}|q)p(x_{3}|u_{1},u_{2},q)p(y|x_{1}%
,x_{2},x_{3})\label{pmf}%
\end{equation}
for some auxiliary random variables $U_1$, $U_2$ and $Q$.
\end{prop}

\begin{proof}
A more general MAC model with three users and any combination of
\textquotedblleft common messages\textquotedblright\ (i.e., messages known
\textquotedblleft cognitively\textquotedblright\ to more than one user) is
studied in Sec. VII of \cite{Slepian}, from which Proposition \ref{prop:1} can
be obtained as a special case. However, since a proof is not provided in
\cite{Slepian}, and the technique developed here will be used in deriving
other achievable regions in the paper, we provide a proof in Appendix
\ref{app:1}.
\end{proof}

Towards the goal of accounting for non-ideal connections between sources and
relay (as in the original cMACr), we next consider the cases of partial and
limited-rate cognition (rigorously defined below). We start with the
\textit{partial cognition }model, in which the relay is informed of the
message of only one of the two users, say of message $W_{1}$.

\begin{prop}
\label{prop:2} The capacity region of the MAC with a partially cognitive relay
(informed only of the message $W_{1})$ is given by the closure of the set of
all non-negative $(R_{1},R_{2},R_{3})$ satisfying
\begin{subequations}
\label{partially cognitive}%
\begin{align}
R_{2} &  \leq I(X_{2};Y|X_{1},X_{3},Q), \\
R_{3} &  \leq I(X_{3};Y|X_{1},X_{2},Q), \\
R_{1}+R_{3} &  \leq I(X_{1},X_{3};Y|X_{2},Q), \\
R_{2}+R_{3} &  \leq I(X_{2},X_{3};Y|X_{1},Q),
\end{align}
and
\begin{align}
R_{1}+R_{2}+R_{3} &  \leq I(X_{1},X_{2},X_{3};Y|Q).
\end{align}
for an input distribution of the form $p(q)p(x_{2}|q)p(x_{1},x_{3}|q).$
\end{subequations}
\end{prop}

\begin{proof}
The proof can be found in Appendix \ref{app:2}.
\end{proof}


\begin{rem}
The capacity region characterization requires two auxiliary random variables
in Proposition \ref{prop:1} (and in \cite{Slepian}), while no auxiliary random
variables are required in the formulation of Proposition \ref{prop:2}. This is
because, in the scenario covered by Proposition \ref{prop:1}, the relay's
codeword can depend on both $W_{1}$ and $W_{2}$, and the auxiliary random
variables quantify the amount of dependence on each message. On the contrary,
for Proposition \ref{prop:2}, the relay cooperates with only one source, and
no auxiliary random variable is needed. To further elaborate on this point,
another special case of the channel in Fig. \ref{model2} in which no auxiliary
random variable is necessary to achieve the capacity region is obtained when
each transmitter is connected to the receiver via an orthogonal channel, i.e.,
we have $Y=(Y_{1},Y_{2},Y_{3})$ and $p(y_{1},y_{2},y_{3}|x_{1},x_{2}%
,x_{3})=\prod_{i=1}^{3}p(y_{i}|x_{i})$. In this case, unlike Proposition
\ref{prop:2}, the lack of auxiliary random variables reflects the fact that no
coherent combining gain can be accrued via the use of the relay due to the
channels' orthogonality. Defining $C_{i}=\max_{p(x_{i})}I(X_{i};Y_{i})$, for
$i=1,2,3$, we obtain from Proposition \ref{prop:1} that the capacity region is
given by $\{(R_{1},R_{2},R_{3}):0\leq R_{1},0\leq R_{2},0\leq R_{3}\leq
C_{3},R_{1}+R_{2}\leq C_{1}+C_{2},R_{2}+R_{3}\leq C_{2}+C_{3},R_{1}%
+R_{2}+R_{3}\leq C_{1}+C_{2}+C_{3}\}$.
\end{rem}

The model in Fig. \ref{model2} can be further generalized to a scenario with
\textit{limited-capacity cognition}, in which the sources are connected to the
relay via finite-capacity orthogonal links, rather than having a priori
knowledge of the terminals' messages. This channel can be seen as an
intermediate step between the MAC with cognitive relay studied above and the
multiple access channel with relay for which an achievable region was derived
in \cite{Kramer:IT:05} for the case $R_{3}=0$. In particular, assume that
terminal 1 can communicate with the relay, prior to transmission, via a link
of capacity\ $C_{1}$ and that similarly terminal 2 can communicate with the
relay via a link of capacity $C_{2}.$ The following proposition establishes
the capacity of such a channel.

\begin{prop}\label{prop:conf}
The capacity region of the MAC with a cognitive relay connected to the source
terminals via (unidirectional) links of capacities $C_{1}$ and $C_{2}$ is
given by
\begin{subequations}
\label{conf}%
\begin{align}
R_{1} &  \leq I(X_{1};Y|X_{2},X_{3},U_{1},U_{2},Q)+C_{1}, \\
R_{2} &  \leq I(X_{2};Y|X_{1},X_{3},U_{1},U_{2},Q)+C_{2}, \\
R_{3} &  \leq I(X_{3};Y|X_{1},X_{2},U_{1},U_{2},Q), \\
R_{1}+R_{2} &  \leq I(X_{1},X_{2};Y|X_{3},U_{1},U_{2},Q)+C_{1}+C_{2}, \\
R_{1}+R_{3} &  \leq\min\left\{
\begin{array}
[c]{l}%
I(X_{1},X_{3};Y|X_{2},U_{1},U_{2},Q)+C_{1},\text{ }\\
I(X_{1},X_{3};Y|X_{2},U_{2},Q)
\end{array}
\right\}, \\
R_{2}+R_{3} &  \leq\min\left\{
\begin{array}
[c]{l}%
I(X_{2},X_{3};Y|X_{1},U_{1},U_{2},Q)+C_{2}\\
I(X_{2},X_{3};Y|X_{1},U_{1},Q)
\end{array}
\right\}
\end{align}
and
\begin{align}
R_{1}+R_{2}+R_{3} &  \leq\min\left\{
\begin{array}
[c]{l}%
I(X_{1},X_{2},X_{3};Y|U_{1},U_{2},Q)+C_{1}+C_{2},\\
I(X_{1},X_{2},X_{3};Y|U_{1},Q)+C_{1},\\
I(X_{1},X_{2},X_{3};Y|U_{2},Q)+C_{2},\\
I(X_{1},X_{2},X_{3};Y|Q)
\end{array}
\right\}
\end{align}
for some auxiliary random variables $U_1$, $U_2$ and $Q$ with joint distribution of the form (\ref{pmf}).
\end{subequations}
\end{prop}

\begin{proof}
The proof is sketched in Appendix \ref{app:conf}.
\end{proof}

\begin{rem}\label{r:cMAC_cogrel}
Based on the results of this section, we can now make a further step towards
the analysis of the cMACr of Fig. \ref{model} by considering the \textit{cMACr
with a cognitive relay}. This channel is given as in Fig. \ref{model} with the
only difference that the relay here is informed \textquotedblleft for
free\textquotedblright\ of the messages $W_{1}$ and $W_{2}$ (similarly to Fig.
\ref{model2}) and that the signal received at the relay is non-informative,
e.g., $\mathcal{Y}_{2}=\emptyset.$ The capacity of such a channel follows
easily from Proposition \ref{prop:1} by taking the union over the distribution
$p(q)p(x_{1},u_{1}|q)p(x_{2},u_{2}|q)p(x_{3}|u_{1},u_{2},q)$ $p(y_{1}%
,y_{2}|x_{1},x_{2},x_{3})$ of the intersection of the two rate regions
(\ref{region cognitive}) evaluated for the two outputs $Y_{1}$ and $Y_{2}$.
Notice that this capacity region depends on the channel inputs only through
the marginal distributions $p(y_{1}|x_{1},x_{2},x_{3})$ and $p(y_{2}%
|x_{1},x_{2},x_{3}).$
\end{rem}


\section{Inner and Outer bounds on the Capacity Region of the Compound MAC
with a Relay}\label{s:cMACr}

In this section, we focus on the general cMACr model illustrated in Fig.
\ref{model}. As stated in Section \ref{s:intro}, single-letter
characterization of the capacity region for this model is open even for
various special cases. Our goal here is to provide inner and outer bounds,
which are then shown to be tight in certain meaningful special scenarios.

The following inner bound is obtained by the decode-and-forward (DF) strategy
\cite{Cover:IT:79} at the relay terminal. The relay fully decodes both
messages of both users so that we have a MAC from the transmitters to the
relay terminal. Once the relay has decoded the messages, the transmission to
the receivers takes place similarly to the MAC with a cognitive relay model of
Section \ref{s:MAC_cogrel}.

\begin{prop}\label{p:achDF}
For the cMACr as seen in Fig. \ref{model}, any rate triplet $(R_{1}%
,R_{2},R_{3})$ with $R_{j}\geq0$, $j=1,2,3$, satisfying
\begin{subequations}
\label{ach:DF}%
\begin{align}
R_{1} &  \leq I(X_{1};Y_{3}|U_{1}, X_{2},X_{3},Q), \label{ach 1}\\
R_{2} &  \leq I(X_{2};Y_{3}|U_{2}, X_{1},X_{3},Q), \\
R_{1}+R_{2} &  \leq I(X_{1},X_{2};Y_{3}|U_{1}, U_{2}, X_{3},Q), \label{ach 3}\\
R_{3} &  \leq\min\{I(X_{3};Y_{1}|X_{1},X_{2},U_{1},U_{2},Q),\text{ }%
I(X_{3};Y_{2}|X_{1},X_{2},U_{1},U_{2},Q)\}, \label{ach 4}\\
R_{1}+R_{3} &  \leq\min\{I(X_{1},X_{3};Y_{1}|X_{2},U_{2},Q),\text{ }%
I(X_{1},X_{3};Y_{2}|X_{2},U_{2},Q)\}, \\
R_{2}+R_{3} &  \leq\min\{I(X_{2},X_{3};Y_{1}|X_{1},U_{1},Q),\text{ }%
I(X_{2},X_{3};Y_{2}|X_{1},U_{1},Q)\}
\end{align}
and
\begin{align}
R_{1}+R_{2}+R_{3} &  \leq\min\{I(X_{1},X_{2},X_{3};Y_{1}|Q),\text{ }%
I(X_{1},X_{2},X_{3};Y_{2}|Q)\}\label{ach 7}%
\end{align}
for auxiliary random variables $U_1$, $U_2$ and $Q$ with a joint distribution of the form
\end{subequations}
\begin{equation}
p(q)p(x_{1},u_{1}|q)p(x_{2},u_{2}|q)p(x_{3}|u_{1},u_{2},q)p(y_{1},y_{2}%
,y_{3}|x_{1},x_{2},x_{3})
\end{equation}
is achievable by DF.
\end{prop}

\begin{proof}
The proof follows by combining the block-Markov transmission strategy with DF
at the relay studied in Sec. IV-D of \cite{Kramer:IT:05}, the joint encoding
used in Proposition \ref{prop:1} to handle the private relay message and
backward decoding at the receivers. Notice that conditions (\ref{ach 1}%
)-(\ref{ach 3}) ensure correct decoding at the relay, whereas (\ref{ach 4}%
)-(\ref{ach 7}) follow similarly to Proposition \ref{prop:1} and Remark
\ref{r:cMAC_cogrel} ensuring correct decoding of the messages at both receivers.
\end{proof}

Next, we consider applying the compress-and-forward (CF) strategy
\cite{Cover:IT:79} at the relay terminal. With CF, the relay does not decode
the source message, but facilitates decoding at the destination by
transmitting a quantized version of its received signal. In quantizing its
received signal, the relay takes into consideration the correlated received
signal at the destination terminal and applies Wyner-Ziv source compression
(see \cite{Cover:IT:79} for details). In the cMACr scenario, unlike the
single-user relay channel, we have two distinct destinations, each with
different side information correlated with the relay received signal. This
situation is similar to the problem of lossy broadcasting of a common source
to two receivers with different side information sequences considered in
\cite{Nayak:IT:08} (and solved in some special cases), and applied to the
two-way relay channel setup in \cite{Gunduz:All:08}. Here, for simplicity, we
consider broadcasting only a single quantized version of the relay received
signal to both receivers. The following proposition states the corresponding
achievable rate region.

\begin{prop}
\label{p:achCF} For the cMACr of Fig. \ref{model}, any rate triplet
$(R_{1},R_{2},R_{3})$ with $R_{j}\geq0$, $j=1,2,3$, satisfying
\begin{align}
R_{1}  &  \leq\min\{I(X_{1};Y_{1},\hat{Y}_{3}|X_{2},X_{3},Q),I(X_{1}%
;Y_{2},\hat{Y}_{3}|X_{2},X_{3},Q)\}, \label{a:CF1}\\
R_{2}  &  \leq\min\{I(X_{2};Y_{2},\hat{Y}_{3}|X_{1},X_{3},Q),I(X_{2}%
;Y_{1},\hat{Y}_{3}|X_{1},X_{3},Q)\}, \label{a:CF2}%
\end{align}
and
\begin{align}
R_{1}+R_{2}\leq\min\{I(X_{1},X_{2};Y_{1},\hat{Y}_{3}|X_{3},Q),I(X_{1}%
,X_{2};Y_{2},\hat{Y}_{3}|X_{3},Q)\}\label{a:CF3}%
\end{align}
such that
\begin{align}
R_{3}+I(Y_{3};\hat{Y}_{3}|X_{3},Y_{1},Q)  &  \leq I(X_{3};Y_{1}%
|Q)\label{a:cCF4}
\end{align}
and
\begin{align}
R_{3}+I(Y_{3};\hat{Y}_{3}|X_{3},Y_{2},Q)  &  \leq I(X_{3};Y_{2}%
|Q)\label{a:cCF5}%
\end{align}
for random variables $\hat{Y}_{3}$ and $Q$ satisfying the joint distribution
\begin{equation}
p(q,x_{1},x_{2},x_{3},y_{1},y_{2},y_{3},\hat{y}_{3})=p(q)p(x_{1}%
|q)p(x_{2}|q)p(x_{3}|q)p(\hat{y}_{3}|y_{3},x_{3},q)p(y_{1},y_{2},y_{3}%
|x_{1},x_{2},x_{3})
\end{equation}
is achievable with $\hat{Y}_{3}$ having bounded cardinality.
\end{prop}

\begin{proof}
The proof can be found in Appendix \ref{app:achCF}.
\end{proof}

\begin{rem}
The achievable rate region given in Proposition \ref{p:achCF} can be
potentially improved. Instead of broadcasting a single quantized version of
its received signal, the relay can transmit two descriptions so that the
receiver with an overall better quality in terms of its channel from the relay
and the side information received from its transmitter, receives a better
description, and hence higher rates (see \cite{Nayak:IT:08} and
\cite{Gunduz:All:08} for details). Another possible extension which we will
not pursue here is to use the partial DF scheme together with the above CF
scheme similar to the coding technique in \cite{Gunduz:All:08}.
\end{rem}

We are now interested in studying the special case in which each source
terminal can reach only one of the destination terminals directly. Assume, for
example, that there is no direct connection between source terminal 1 and
destination terminal 2, and similarly between source terminal 2 and
destination terminal 1. In practice, this setup might model either a larger
distance between the disconnected terminals, or some physical constraint in
between the terminals blocking the connection. Obviously, in such a case, no
positive multicasting rate can be achieved without the help of the relay, and
hence, the relay is essential in providing coverage to multicast data to both
receivers. We model this scenario by the following (symbol-by-symbol) Markov
chain conditions:
\begin{subequations}
\label{markov}%
\begin{align}
Y_{1}- &  (X_{1},X_{3})-X_{2}\mbox{ and }\label{ch_ass_1}\\
Y_{2}- &  (X_{2},X_{3})-X_{1},\label{ch_ass_2}%
\end{align}
which state that the output at receiver 1 depends only on the inputs of
transmitter 1 and the relay (\ref{ch_ass_1}), and similarly, the output at
receiver 2 depends only on the inputs of transmitter 2 and the relay
(\ref{ch_ass_2}). The following proposition provides an outer bound for the
capacity region in such a scenario.
\end{subequations}
\begin{prop}\label{p:outerbound}
Assuming that the Markov chain conditions (\ref{markov}) hold for any channel
input distribution satisfying (\ref{distribution 1}), a rate triplet
($R_{1},R_{2},R_{3}$) with $R_{j}\geq0$, $j=1,2,3,$ is achievable only if
\begin{subequations}
\label{outer bound}%
\begin{align}
R_{1} &  \leq I(X_{1};Y_{3}|U_{1}, X_{2}, X_{3},Q), \\
R_{2} &  \leq I(X_{2};Y_{3}|U_{2}, X_{1},X_{3},Q), \\
R_{3} &  \leq\min\{I(X_{3};Y_{1}|X_{1},X_{2},U_{1},U_{2},Q),\text{ }%
I(X_{3};Y_{2}|X_{1},X_{2},U_{1},U_{2},Q)\}, \\
R_{1}+R_{3} &  \leq\min\{I(X_{1},X_{3};Y_{1}|U_{2},Q),\text{ }I(X_{3}%
;Y_{2}|X_{2},U_{2},Q)\}, \\
R_{2}+R_{3} &  \leq\min\{I(X_{3};Y_{1}|X_{1},U_{1},Q),\text{ }I(X_{2}%
,X_{3};Y_{2}|U_{1},Q)\}
\end{align}
and
\begin{align}
R_{1}+R_{2}+R_{3} &  \leq\min\{I(X_{1},X_{3};Y_{1}|Q),\text{ }I(X_{2}%
,X_{3};Y_{2}|Q)\}
\end{align}
for some auxiliary random variables $U_1$, $U_2$ and $Q$ satisfying the joint distribution
\end{subequations}
\begin{equation}
p(q)p(x_{1},u_{1}|q)p(x_{2},u_{2}|q)p(x_{3}|u_{1},u_{2},q)p(y_{1},y_{2}%
,y_{3}|x_{1},x_{2},x_{3}). \label{distribution outer}%
\end{equation}
\end{prop}

\begin{proof}
The proof can be found in Appendix \ref{app:outbnd}.
\end{proof}

By imposing the condition (\ref{markov}) on the DF achievable rate region of
Proposition \ref{p:achDF}, it can be easily seen that the only
difference between the outer bound (\ref{outer bound}) and the achievable
region with DF (\ref{ach:DF}) is that the latter contains the additional
constraint (\ref{ach 3}), which generally reduces the rate region. The
constraint (\ref{ach 3}) accounts for the fact that the DF scheme leading to
the achievable region (\ref{ach:DF}) prescribes both messages $W_{1}$ and
$W_{2}$ to be decoded at the relay terminal. The following remark provides two
examples in which the DF scheme achieves the outer bound (\ref{outer bound})
and thus the\ capacity region. In both cases, the multiple access interference
at the relay terminal is eliminated from the problem setup so that the
condition (\ref{ach 3}) does not limit the performance of DF.

\begin{rem}
In addition to the Markov conditions in (\ref{markov}), consider orthogonal
channels from the two users to the relay terminal, that is, we have
$Y_{3}\triangleq(Y_{31},Y_{32})$, where $Y_{3k}$ depends only on inputs
$X_{k}$ and $X_{3}$ for $k=1,2$; that is, we assume $X_{1}-(X_{2}%
,X_{3})-Y_{32}$ and $X_{2}-(X_{1},X_{3})-Y_{31}$ form Markov chains for any
input distribution. Then, it is easy to see that the sum-rate constraint at
the relay terminal is redundant and hence the outer bound in Proposition
\ref{p:outerbound} and the achievable rate region with DF in Proposition
\ref{p:achDF} match, yielding the full capacity region for this
scenario. As another example where DF\ is optimal, we consider a \textit{relay
multicast channel} setup, in which a single relay helps transmitter 1 to
multicast its message $W_{1}$ to both receivers, i.e., $R_{2}=R_{3}=0$ and
$X_{2}=\emptyset$. For such a setup, under the assumption that $X_{1}%
-X_{3}-Y_{2}$ forms a Markov chain, the achievable rate with DF relaying in
Proposition \ref{p:achDF} and the above outer bound match.
Specifically, the capacity $C$ for this \textit{multicast relay channel} is
given by
\begin{equation}
C=\max_{p(x_{1},x_{3})}\min\{I(X_{1};Y_{3}|X_{3}),\text{ }I(X_{1},X_{3}%
;Y_{1}),\text{ }I(X_{3};Y_{2})\}.
\end{equation}

\end{rem}

Notice that, apart from some special cases (like the once illustrated above),
the achievable rate region with DF is in general suboptimal due to the
requirement of decoding both the individual messages at the relay terminal. In
fact, this requirement may be too restrictive, and simply decoding a function
of the messages at the relay might suffice. To illustrate this point, consider
the special case of the cMACr characterized by $X_{i}=(X_{i,1},X_{i,2})$,
$Y_{i}=(Y_{i,1},Y_{i,2})$ and $Y_{i,1}=X_{i,1}$ for $i=1,2$ and the channel
given as
\[
p(y_{1},y_{2},y_{3})=p(y_{3}|x_{1,2},x_{2,2})p(y_{1,1}|x_{1,1})p(y_{2,1}%
|x_{2,1})p(y_{1,2},y_{2,2}|x_{3}).
\]
In this model, each transmitter has an error-free orthogonal channel to its
receiver. By further assuming that these channels have enough capacity to
transmit the corresponding messages reliably (i.e., message $i$ is available
at receiver $i)$, the channel at hand is seen to be a form of the two-way relay
channel. In this setup, as shown in \cite{Knopp:GDR:07}, \cite{Wilson:IT:08},
\cite{Nam:IZS:08} and \cite{Gunduz:All:08}, DF relaying is suboptimal while
using a structured code achieves the capacity in the case of finite field
additive channels and improves the achievable rate region in the case of
Gaussian channels. In the following section, we explore a similar scenario for
which the outer bound (\ref{outer bound}) is the capacity region of the cMACr,
which cannot be achieved by either DF or CF.

\section{Binary cMACr: Achieving Capacity Through Structured Codes
\label{sec:linear binary}}

Random coding arguments have been highly successful in proving the existence
of capacity-achieving codes for many source and channel coding problems in
multi-user information theory, such as MACs, BCs, RCs with degraded signals
and Slepian-Wolf source coding. However, there are various multi-user
scenarios for which the known random coding-based achievability results fail
to achieve the capacity, while \emph{structured codes} can be shown to perform
optimally. The best known such example for such a setup is due to K\"{o}rner
and Marton \cite{Korner:IT:79}, who considered encoding the modulo sum of two
binary random variables. See \cite{Nazer:ETT:08} for more examples and references.

Here, we consider a binary symmetric (BS) cMACr model and show that structured
codes achieve its capacity, while the rate regions achievable with DF or CF
schemes are both suboptimal. We model the BS cMACr as follows:
\begin{subequations}
\label{bsc}%
\begin{align}
Y_{1} &  =X_{1}\oplus X_{3}\oplus Z_{1},\\
Y_{2} &  =X_{2}\oplus X_{3}\oplus Z_{2},\mbox { and }\\
Y_{3} &  =X_{1}\oplus X_{2}\oplus Z_{3}%
\end{align}
\end{subequations}
where $\oplus$ denotes binary addition, and the noise components $Z_{i}$ are
independent identically distributed (i.i.d.) with\footnote{$X\sim
\mathcal{B}(\varepsilon)$ denotes a Bernoulli distribution for which
$p(X=1)=\varepsilon$ and $p(X=0)=1-\varepsilon$.} $\mathcal{B}(\varepsilon
_{i})$, $i=1,2,3$, and they are independent of each other and the channel
inputs. Notice that this channel satisfies the Markov condition given in
(\ref{markov}). We assume that the relay does not have a private message,
i.e., $R_{3}=0$. The capacity region for this BS cMACr, which can be achieved
by structured codes, is characterized in the following proposition.

\begin{prop}\label{p:binary}
For the binary symmetric cMACr characterized in (\ref{bsc}), the capacity region
is the union of all rate pairs $(R_{1}, R_{2})$ satisfying
\begin{subequations}
\label{binary example}%
\begin{align}
R_{1}  &  \leq1-H_{b}(\varepsilon_{3}),\label{bin example 1}\\
R_{2}  &  \leq1-H_{b}(\varepsilon_{3})~~ \mbox { and }\label{bin example 2}\\
R_{1}+R_{2}  &  \leq\min\{1-H_{b}(\varepsilon_{1}),1-H_{b}(\varepsilon
_{2})\},\label{bin example 3}%
\end{align}
\end{subequations}
where $H_{b}(\varepsilon)$ is the binary entropy function defined as
$H_{b}(\varepsilon) \triangleq-\varepsilon\log\varepsilon-(1-\varepsilon
)\log(1-\varepsilon).$
\end{prop}

\begin{proof}
The proof can be found in Appendix \ref{app:binary}.
\end{proof}

For comparison, the rate region achievable with the DF scheme given in
(\ref{ach:DF}) is given by (\ref{binary example}) with the additional
constraint
\[
R_{1}+R_{2}\leq1-H_{b}(\varepsilon_{3}),
\]
showing that the DF scheme achieves the capacity (\ref{binary example}) only
if $\varepsilon_{3}\leq\min\{\varepsilon_{1},\varepsilon_{2}\}$. The
suboptimality of DF\ follows from the fact that the relay terminal needs
to decode only the binary sum of the messages, rather than the individual messages
sent by the source terminals. In fact, in the achievability scheme leading to
(\ref{binary example}), the binary sum is decoded at the relay and broadcast
to the receivers, which can then decode both messages using this binary sum.

\section{Gaussian Channels}\label{s:Gaussian}

In this section, we focus on the Gaussian channel setup and find the Gaussian
counterparts of the rate regions characterized in Section \ref{s:MAC_cogrel}
and Section \ref{s:cMACr}. We will also quantify the gap between the
inner and outer bounds for the capacity region of the cMACr proposed in
Section \ref{s:cMACr}. As done in\ Sec. \ref{s:MAC_cogrel}, we first
deal with the MAC with a cognitive relay model.

\subsection{Gaussian MAC with a Cognitive Relay}

We first consider the Gaussian MAC with a cognitive relay setup. The multiple
access channel at time $i$, $i=1,\ldots,n$, is characterized by the relation
\begin{equation}
Y_{i}=X_{1i}+X_{2i}+X_{3i}+Z_{i},\label{GMAC}%
\end{equation}
where $Z_{i}$ is the channel noise at time $i$, which is i.i.d. zero-mean
Gaussian with unit variance. We impose a separate average block power
constraint on each channel input:
\begin{equation}
\frac1{n}\sum\limits_{i=1}^{n}E[X_{ji}^{2}]\leq P_{j}\label{power constraints}%
\end{equation}
for $j=1,2,3$. The capacity region for this Gaussian model can be
characterized as follows.

\begin{prop}
\label{prop:6} The capacity region of the Gaussian MAC with a cognitive relay
(\ref{GMAC}) with power constraints (\ref{power constraints}) is the union of
all rate triplets $(R_{1},R_{2},R_{3})$ satisfying
\begin{subequations}
\label{G cognitive}%
\begin{align}
R_{3} &  \leq\frac12\log(1+(1-\alpha_{3}^{\prime}-\alpha_{3}^{\prime\prime
})P_{3}), \\
R_{1}+R_{3} &  \leq\frac12\log(1+P_{1}+(1-\alpha_{3}^{\prime\prime}%
)P_{3}+2\sqrt{\alpha_{3}^{\prime}P_{1}P_{3}}), \\
R_{2}+R_{3} &  \leq\frac12\log(1+P_{2}+(1-\alpha_{3}^{\prime})P_{3}%
+2\sqrt{\alpha_{3}^{\prime\prime}P_{2}P_{3}}))
\end{align}
and
\begin{align}
R_{1}+R_{2}+R_{3} &  \leq\frac12\log(1+P_{1}+P_{2}+P_{3}+2\sqrt{\alpha
_{3}^{\prime}P_{1}P_{3}}+2\sqrt{\alpha_{3}^{\prime\prime}P_{2}P_{3}}),
\end{align}
where the union is taken over all parameters $0\leq\alpha_{3}^{\prime}%
,\alpha_{3}^{\prime\prime}\leq1$ and $\alpha_{3}^{\prime}+\alpha_{3}%
^{\prime\prime}\leq1.$
\end{subequations}
\end{prop}

\begin{proof}
The proof can be found in Appendix \ref{app:prop6}.
\end{proof}

Notice that $\alpha_{3}^{\prime}$ and $\alpha_{3}^{\prime\prime}$ in
(\ref{G cognitive}) represent the fraction of total power invested by the
cognitive relay to help transmitter 1 and transmitter 2, respectively. Next,
we present the capacity region for the Gaussian partially cognitive relay
setup of (\ref{partially cognitive}).

\begin{prop}
The capacity region of the Gaussian MAC with a partially cognitive relay
(informed only of the message $W_{1})$ is given by
\begin{subequations}
\label{G partial}%
\begin{align}
R_{2} &  \leq\frac12\log(1+P_{2}), \\
R_{3} &  \leq\frac12\log(1+(1-\rho^{2})P_{3}), \\
R_{1}+R_{3} &  \leq\frac12\log(1+P_{1}+P_{3}+2\rho\sqrt{P_{1}P_{3}}), \\
R_{2}+R_{3} &  \leq\frac12\log(1+P_{2}+(1-\rho^{2})P_{3})
\end{align}
and
\begin{align}
R_{1}+R_{2}+R_{3} &  \leq\frac12\log(1+P_{1}+P_{2}+P_{3}+2\rho\sqrt{P_{1}%
P_{3}})\label{sum-rate partial}%
\end{align}
with the union taken over the parameter $0\leq\rho\leq1.$
\end{subequations}
\end{prop}

\begin{proof}
The result follows straightforwardly from (\ref{partially cognitive}) and the conditional maximum
entropy theorem by defining $\rho$ as the correlation coefficient between
$X_{1}$ and $X_{3}.$
\end{proof}

Notice that, the same arguments as above can also be extended to the MAC with cognition via
finite-capacity links of Proposition \ref{prop:conf}.

\subsubsection{Numerical Examples}

For clarity of the presentation we consider $R_{3}=0.$ In this case, it is
clear that the choice $\alpha_{3}^{\prime}+\alpha_{3}^{\prime\prime}=1$ is
optimal for (\ref{G cognitive}) and $\rho=1$ is optimal in (\ref{G partial}).
Fig. \ref{regionrelay} shows the capacity regions with full or partial
cognition ((\ref{G cognitive}) and (\ref{G partial}), respectively) for
$P_{1}=P_{2}=3~\mathrm{dB}$ and for different values of $P_{3}$, namely
$P_{3}=-6~\mathrm{dB}$ and $3~\mathrm{dB}$. It can be observed from Fig.
\ref{regionrelay} that, even with a small power $P_{3}$, a cognitive relay has
the potential for significantly improving the achievable rate regions.
Moreover, in the partially cognitive case, this advantage is accrued not only
by the transmitter that directly benefits from cognition (here transmitter 1)
but also by the other transmitter (transmitter 2), due to the fact that
cognition is able to boost the achievable sum-rate (see
(\ref{sum-rate partial})).

\begin{figure}[ptb]
\begin{center}
\includegraphics[width=4.2644in]
{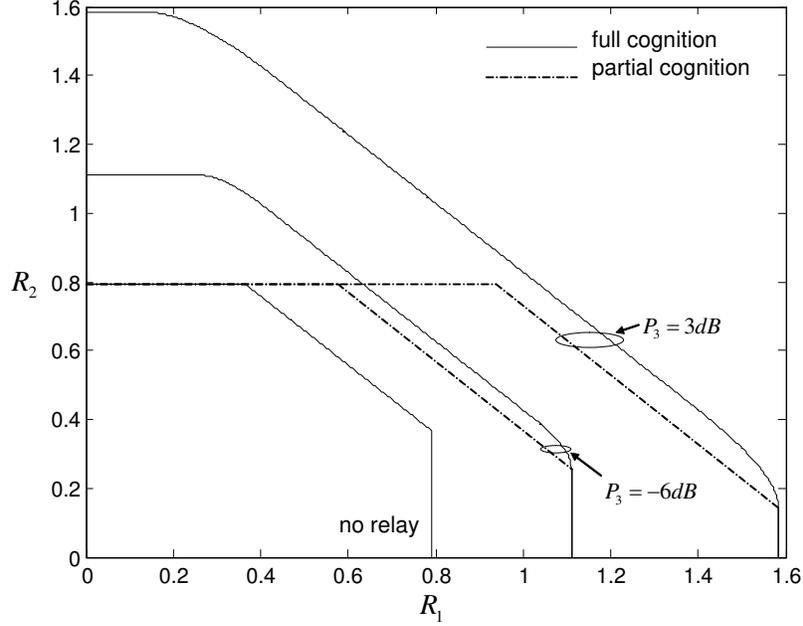}
\end{center}
\caption{Capacity regions of the Gaussian MAC with a cognitive relay with full
or partial cognition ((\ref{G cognitive}) and (\ref{G partial}), respectively)
for $P_{1}=P_{2}=3~\mathrm{dB}$ and for different values of $P_{3}$, namely
$P_{3}=-6~\mathrm{dB}$ and $3~\mathrm{dB.}$}
\label{regionrelay}
\end{figure}

We now consider a typical cognitive radio scenario where the two
\textquotedblleft primary\textquotedblright\ users, transmitter 1 and
transmitter 2, transmit at rates $R_{1}$ and $R_{2}$, respectively, within the
standard MAC capacity region with no relay (i.e., ($R_{1},R_{2}$) satisfy
(\ref{G cognitive}) with $R_{3}=0$ and $P_{3}=0)$ and are oblivious to the
possible presence of a cognitive node transmitting to the same receiver$.$ By
assumption, the cognitive node can rapidly acquire the messages of the two
active primary users (exploiting the better channel from the primary users as
compared to the receiver) and is interested in transmitting at the maximum
rate $R_{3}$ that does not affect the rates achievable by the primary users.
In other words, the rate $R_{3}$ is selected so as to maximize $R_{3}$ under
the constraint that $(R_{1},R_{2},R_{3})$ still belongs to the capacity region
(the one characterized by (\ref{G cognitive}) for full cognition and by
(\ref{G partial}) for partial cognition). Fig. \ref{maxsecrate} shows such a rate
$R_{3} $ for both full and partial cognitive relays for $P_{1}=P_{2}=3dB$ and
two different primary rate pairs, namely $R_{1}=R_{2}=0.3$ and $R_{1}%
=R_{2}=0.55$ (which is close to the sum-rate boundary as shown in Fig.
\ref{region cognitive}). It is seen that both full and partial cognition
afford remarkable achievable rates even when the primary users select rates at
the boundary of their allowed rates.

\begin{figure}[ptb]
\begin{center}
\includegraphics[width=4.7928in]{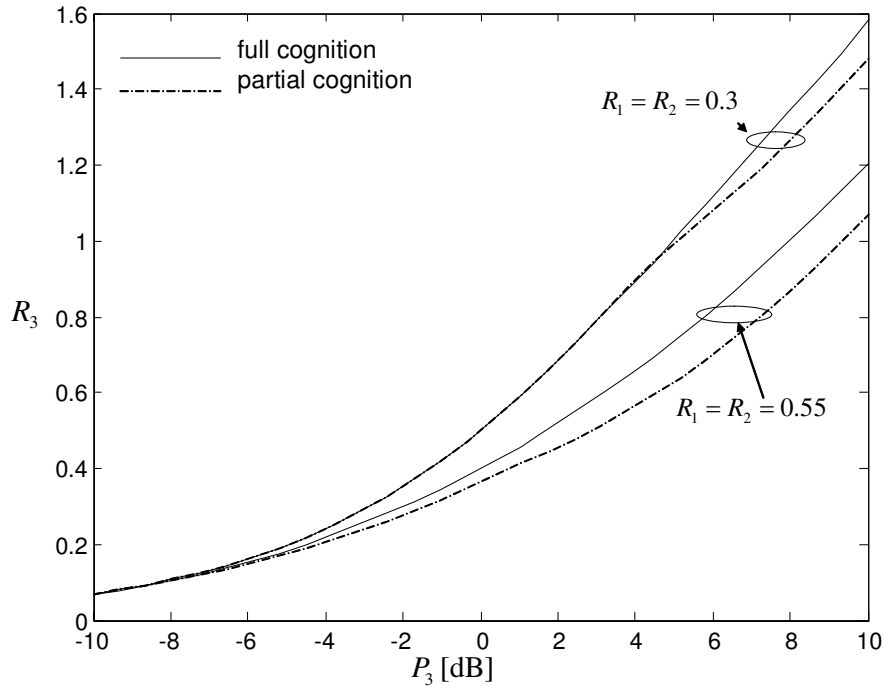}
\end{center}
\caption{Maximum rate $R_{3}$ that does not affect the rates achievable by the
primary users 1 and 2 for $P_{1}=P_{2}=3dB$ and $R_{1}=R_{2}=0.3$ or
$R_{1}=R_{2}=0.55$ ($R_{3}$ is the maximum relay rate so that $(R_{1}%
,R_{2},R_{3})$ still belongs to the capacity region ((\ref{G cognitive}) for
full cognition and ( \ref{G partial}) for partial cognition). }%
\label{maxsecrate}
\end{figure}

\subsection{Gaussian Compound MAC with a Relay}\label{p:GcMACr}

A Gaussian cMACr\ satisfying the Markov conditions (\ref{markov}) is given by
\begin{subequations}
\label{G mac relay}%
\begin{align}
Y_{1} &  =X_{1}+\eta X_{3}+Z_{1}\\
Y_{2} &  =X_{2}+\eta X_{3}+Z_{2}\\
Y_{3} &  =\gamma(X_{1}+X_{2})+Z_{3},
\end{align}
where $\gamma\geq0$ is the channel gain from the users to the relay and
$\eta\geq0$ is the channel gain from the relay to both receiver 1 and receiver
2. The noise components $Z_{i}$, $i=1,2,3$ are i.i.d. zero-mean unit variance
Gaussian random variables. We enforce the average power constraints given in
(\ref{power constraints}). Considering for simplicity the case $R_{3}=0,$ we
have the following result.
\end{subequations}
\begin{prop}
The following rate region is achievable for the Gaussian cMACr characterized
by (\ref{G mac relay}) by using the DF strategy:
\begin{subequations}
\label{region compound mac}%
\begin{align}
R_{1} &  \leq\min\left\{
\begin{array}
[c]{l}%
\frac12\log\left(  1+\gamma^{2}P_{1}\left(  1-\frac{\alpha_{1}\alpha
_{3}^{\prime}}{1-\alpha_{2}\alpha_{3}^{\prime\prime}}\right)  \right)  ,\\
\frac12\log\left(  1+P_{1}+\eta^{2}P_{3}(1-\alpha_{3}^{\prime\prime})\right)
\end{array}
\right\}, \\
R_{2} &  \leq\min\left\{
\begin{array}
[c]{l}%
\frac12\log\left(  1+\gamma^{2}P_{2}\left(  1-\frac{\alpha_{2}\alpha
_{3}^{\prime\prime}}{1-\alpha_{1}\alpha_{3}^{\prime}}\right)  \right)  ,\\
\frac12\log\left(  1+P_{2}+\eta^{2}P_{3}(1-\alpha_{3}^{\prime})\right)
\end{array}
\right\}
\end{align}
and
\begin{align}
R_{1}+R_{2} &  \leq\min\left\{
\begin{array}
[c]{l}%
\frac12\log\left(  1+\gamma^{2}(P_{1}+P_{2})\left(  1-\frac{(\sqrt{\alpha
_{1}\alpha_{3}^{\prime}P_{1}}+\sqrt{\alpha_{2}\alpha_{3}^{\prime\prime}P_{2}%
})^{2}}{P_{1}+P_{2}}\right)  \right)  ,\\
\frac12\log\left(  1+P_{1}+\eta^{2}P_{3}+2\eta\sqrt{\alpha_{1}\alpha
_{3}^{\prime}P_{1}P_{3}}\right)  ,\\
\frac12\log\left(  1+P_{2}+\eta^{2}P_{3}+2\eta\sqrt{\alpha_{1}\alpha
_{3}^{\prime\prime}P_{2}P_{3}}\right)
\end{array}
\right\}  ,\label{sum rate gaussian}%
\end{align}
with the union taken over the parameters $0\leq\alpha_{1},\alpha_{2},\alpha
_{3}^{\prime},\alpha_{3}^{\prime\prime}\leq1$ and $\alpha_{3}^{\prime}%
+\alpha_{3}^{\prime\prime}\leq1.$ Moreover, an outer bound to the capacity
region is given by (\ref{region compound mac}) without the first sum-rate
constraint in (\ref{sum rate gaussian}).
\end{subequations}
\end{prop}

\begin{proof}
It is enough to prove that jointly Gaussian inputs are sufficient to exhaust
the DF achievable region (\ref{ach:DF})\ and the outer bound
(\ref{outer bound}). This can be done similarly to Proposition \ref{prop:6}.
Then, setting the random variables at hand as (\ref{def Xj})-(\ref{def X3})
(see Appendix \ref{app:prop6}) in (\ref{ach:DF})\ and (\ref{outer bound}) and
after some algebra the result can be derived.
\end{proof}

It is noted that, similarly to (\ref{G cognitive}), in
(\ref{region compound mac}) the parameters $\alpha_{3}^{\prime}$ and
$\alpha_{3}^{\prime\prime}$ represent the fractions of power that the relay
uses to cooperate with transmitter 1 and 2, respectively. Moreover, the first
term in each of the three $\min\{\cdot\}$ functions correspond to the
condition that the relay is able to decode the two messages, while the other
terms refer to constraints on decoding at the two receivers.

Next, we characterize the achievable rate region for the Gaussian setup with
the CF strategy of Proposition \ref{p:achCF}. Here, we assume a Gaussian
quantization codebook without claiming optimality.

\begin{prop}
The following rate region is achievable for the Gaussian cMACr
(\ref{G mac relay}):
\begin{subequations}
\label{ach:CF}%
\begin{align}
R_{1} &  \leq\frac12\log\left(  1+\frac{\gamma^{2}\alpha_{1}P_{1}}{1+N_{q}%
}\right)
\end{align}
and
\begin{align}
R_{2} &  \leq\frac12\log\left(  1+\frac{\gamma^{2}\alpha_{2}P_{2}}{1+N_{q}%
}\right)
\end{align}
where
\end{subequations}
\[
N_{q}=\frac{1+\gamma^{2}(\alpha_{1}P_{1}\alpha_{2}P_{2}+\alpha_{1}P_{1}%
+\alpha_{2}P_{2})+\min\left\{  \alpha_{1}P_{1},\alpha_{2}P_{2}\right\}  }%
{\eta^{2}P_{3}},
\]
for all $0\leq\alpha_{i}\leq1$, $i=1,2$.
\end{prop}


\subsubsection{Using Structured Codes}

In Sec. \ref{sec:linear binary}, we have shown that for a binary additive
compound MAC with a relay, it is optimal to use structured (block linear) codes
rather than conventional unstructured (random) codes. The reason for this
performance advantage is that linear codes, when received by the relay over an
additive channel, enable the latter to decode the sum of the original messages
with no rate loss, without requiring joint decoding of the messages. Here, in
view of the additive structure of the Gaussian channel, we would like to
extend the considerations of Sec. \ref{sec:linear binary} to the scenario at
hand. For simplicity, we focus on a symmetric scenario where $P_{1}%
=P_{2}=P_{3}=P$, $R_{1}=R_{2}=R$ (and $R_{3}=0).$ Under such assumptions, the
outer bound of Proposition \ref{p:GcMACr} sets the following upper bound on
the equal rate $R$ (obtained by setting $\alpha_{3}^{\prime}=\alpha
_{3}^{\prime\prime}=\alpha_{3}$ and $\alpha_{1}=\alpha_{2}=\alpha$ in
(\ref{region compound mac})):%
\begin{align}
R &  \leq\max_{\substack{0\leq\alpha\leq1 \\0\leq\alpha_{3}\leq1}}\min\left\{
\frac{1}{2}\log\left(  1+\gamma^{2}P\left(  \frac{1-2\alpha\alpha_{3}%
}{1-\alpha\alpha_{3}}\right)  \right)  ,\frac{1}{2}\log\left(  1+P+\eta
^{2}P(1-\alpha_{3})\right)  ,\right. \nonumber\\
&  \left.  \frac{1}{4}\log\left(  1+P\left(  1+\eta^{2}+2\eta\sqrt
{\alpha\alpha_{3}}\right)  \right)  \right\}  ,\label{upper bound symm}%
\end{align}
whereas the rate achievable with DF is given by the right hand side
(\ref{upper bound symm}) with an additional term in $\min\{\cdot\}$ given by
$1/4\cdot\log\left(  1+2\gamma^{2}P\left(  1-2\alpha\alpha_{3}\right)
\right)  .$ The rate achievable by CF can be similarly found from
(\ref{ach:CF}) by setting $\alpha_{1}=\alpha_{2}=\alpha$ and maximizing over
$0\leq\alpha\leq1.$

As is well known, the counterpart of binary block codes over binary
additive channels in the case of Gaussian channels is given by lattice codes
which can achieve the Gaussian channel capacity in the limit of infinite block
lengths (see \cite{Erez:IT:04} for further details). A lattice is a discrete
subgroup of the Euclidean space $\mathds{R}^{n}$ with the vector addition
operation, and hence provides us a modulo sum operation at the relay terminal
similar to the binary case.

For the Gaussian cMACr setting given in (\ref{G mac relay}), we use the same
nested lattice code at both transmitters. Similar to the transmission
structure used in the binary setting, we want the relay terminal to
decode only the modulo sum of the messages, where the modulo operation is with
respect to a coarse lattice as in \cite{Wilson:IT:08}, whereas the messages
are mapped to a fine lattice, i.e., we use the nested lattice structure as in
\cite{Erez:IT:04}. The relay terminal then broadcasts the modulo sum of the
message points to both receivers. Each receiver decodes the message from the
transmitter that it hears directly and the modulo sum of the messages from
the relay as explained in Appendix \ref{app:Lattice}. Using these two, each
receiver can also decode the remaining message. We have the following rate
region that can be achieved by the proposed lattice coding scheme.

\begin{prop}\label{p:Lattice}
For the symmetric Gaussian cMACr characterized by
(\ref{G mac relay}), an equal rate $R$ can be achieved using a lattice
encoding/decoding scheme if
\begin{equation}\label{rate_achievable_lattice}
R \leq \min\left\{  \frac{1}{2}\log\left(  \frac{1}{2}+\gamma^{2}P\right)
,\frac{1}{2}\log\left(  1+P\min\{1,\eta^{2}\}\right)  ,\frac{1}{4}%
\log(1+P(1+\eta^{2}))\right\}  .
\end{equation}

\end{prop}

\begin{proof}
The proof can be found in Appendix \ref{app:Lattice}.
\end{proof}

\begin{rem}
Achievability of (\ref{rate_achievable_lattice}), discussed in Appendix
\ref{app:Lattice}, requires transmission at rates corresponding to symmetric
rate point on the boundary of the MAC regions from each transmitter and the
relay to the corresponding receiver. However, here, of the two senders over
each MAC, one sender employs lattice coding (the sources), so that the standard
joint typicality argument fails to prove achievability of these rate points. The
problem is solved by noticing that, even in this scenario, it is
straightforward to operate at the corner points of the MAC region by using
single user encoding and successive decoding. Now, in general, two different
techniques are possible to achieve any boundary rate point by using only
transmission at the corner-point rates, namely time-sharing and rate-splitting
\cite{Rimoldi:IT:96}. In our case, it can be seen that time-sharing would
generally cause a rate reduction with respect to
(\ref{rate_achievable_lattice}), due to the constraint arising from decoding at
the relay. On the contrary, rate-splitting does not have such a drawback: the relay
terminal splits its message and power into two parts and acts as two virtual
users, while single-user coding is applied for each virtual relay user as well
as the message from the transmitter. Since lattice coding achieves the optimal
performance for single user decoding, we can achieve any point on the boundary
of the MAC region.
\end{rem}

\begin{figure}[ptb]
\begin{center}
\includegraphics[width=5in]
{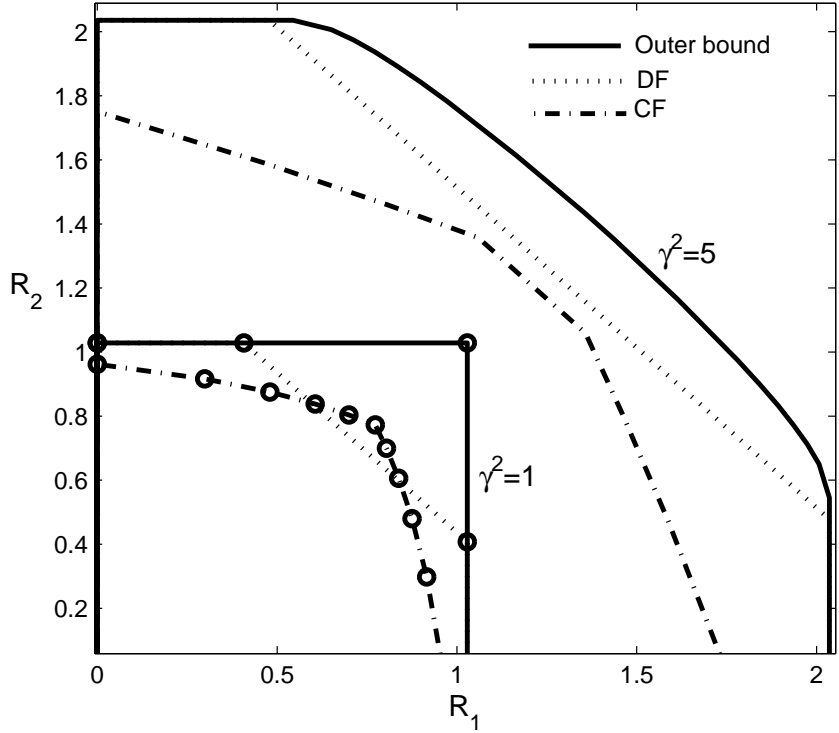}
\end{center}
\caption{Achievable rate region and outer bound for $P_{1}=P_{2}=P_{3}=5~dB$,
$\eta^{2}=10$ and different values of the channel gain from the terminals to
the relay, namely $\gamma^{2}=1,5.$}%
\label{compound region}%
\end{figure}

\subsubsection{Numerical examples}

Consider cMACr with powers $P_{1}=P_{2}=P_{3}=5~dB$ and
channel gain $\eta^{2}=10$ from the relay to the two terminals. Fig.
\ref{compound region} shows the achievable rate region and outer bound for
different values of the channel gain from the terminals to the relay, namely
$\gamma^{2}=1$ and $\gamma^2=5$. It can be seen that, if the channel to the relay is
weak, then CF improves upon DF at certain parts of the rate region. However, as $\gamma^{2}$ increases, DF gets very close to the outer bound dominating the CF rate region, since the sum rate constraint in DF scheme becomes less restricting.

In Fig. \ref{fig:lattice}, the equal rate achievable with lattice codes (\ref{rate_achievable_lattice}) is compared with the upper bound (\ref{upper
bound symm}) and the symmetric rates achievable with DF and CF for $%
\gamma^{2}=1/10$ and $\eta^{2}=10$ versus $P_{1}=P_{2}=P_{3}=P$. We see that, for sufficiently large $P$, the lattice-based scheme is close to optimal, whereas for smaller $P$, CF or DF have better performance. The performance loss of lattice-based schemes with respect to the
upper bound is due to the fact that lattice encoding does not enable
coherent power combining gains at the destination. It is also noted that
both DF and lattice-based schemes have the optimal multiplexing gain of $1/2$
(in terms of equal rate).

\begin{figure}[ptb]
\begin{center}
\includegraphics[width=5in]
{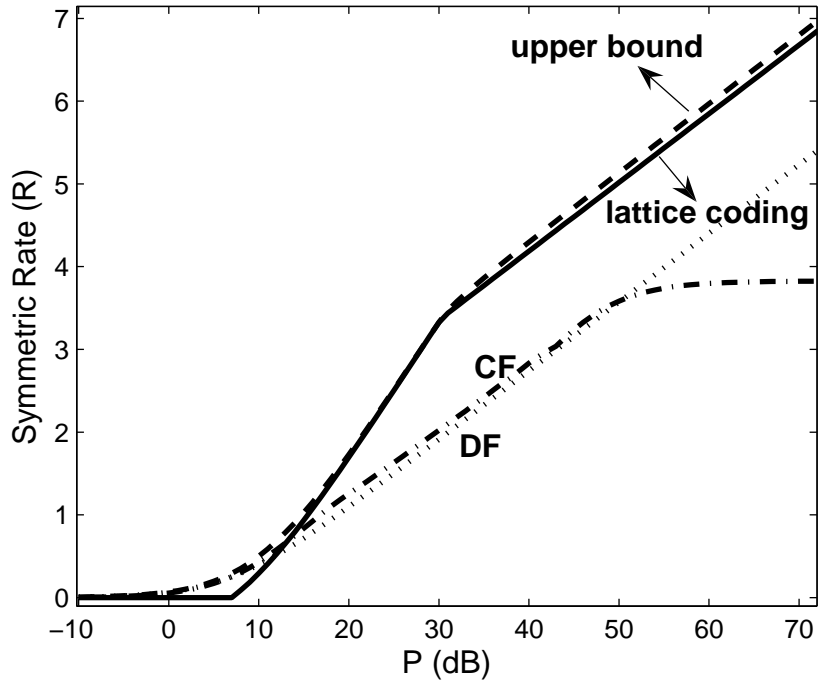}
\end{center}
\caption{Equal rate achievable with lattice codes (\protect\ref%
{rate_achievable_lattice}) compared with the upper bound (\protect\ref{upper
bound symm}) and the rates achievable with DF and CF for $\protect\gamma%
^{2}=1/10$ and $\protect\eta^{2}=10$ versus $P_{1}=P_{2}=P_{3}=P$. }
\label{fig:lattice}
\end{figure}

\section{Conclusions}

\label{s:conc}

We have considered a compound multiple access channel with a relay terminal.
In this model, the relay terminal simultaneously assists both transmitters
while multicasting its own information at the same time. We first have
characterized the capacity region for a multiple access channel with a
cognitive relay and related models of partially cognitive relay and cognition
through finite capacity links. We then have used the coding technique that achieves
the capacity for these models to provide an achievable rate region with DF
relaying in the case of a general cMACr. We have also considered a CF based
relaying scheme, in which the relay broadcasts a compressed version of its
received signal using the received signals at the receivers as side
information. Here we have used a novel joint source-channel coding scheme to
improve the achievable rate region of the underlying multi-user channel coding problem.

We then have focused on another promising approach to improve rates in certain
multi-user networks, namely using structures codes, rather than random coding
schemes. We have proved that the capacity can be achieved by linear block codes in
the case of finite field additive channels. Motivated by the gains achieved
through such structured coding approaches, we have then analyzed the
performance of nested lattice codes in the Gaussian setting. Our results show
that lattice coding achieves rates higher than other random coding schemes for
a wide range of power constraints. We have also presented the achievable rate
regions with the proposed random coding schemes, and provided a comparison.
Our analysis has revealed that no single coding scheme dominates all the
others uniformly over all channel conditions. Hence a combination of various
random coding techniques as well as structured coding might be required to
improve the achievable rates or to meet the upper bounds in a general
multi-user network model.

\appendices

\section{Proof of Proposition \ref{prop:1} \label{app:1}}

\subsection{Types and Typical Sequences}

Here, we briefly review the notions of types and strong typicality that will
be heavily used in the proofs. See \cite{CsiszarKorner} for further details. The type $P_{x^{n}}$ of
an $n$-tuple $x^{n}$ is the empirical distribution
\[
P_{x^{n}}=\frac1{n}N(a|x^{n})
\]
where $N(a|x^{n})$ is the number of occurrences of the letter $a$ in vector
$x^{n}$. The set of all $n$-tuples $x^{n}$ with type $Q$ is called the type
class $Q$ and denoted by $T_{Q}^{n}$. For a probability distribution $p_{X}$, the set of $\epsilon$-strongly typical
$n$-tuples according to $p_{X}$ is denoted by $T_{[X]_{\epsilon}}^{n}$ and is
defined by
\[
T_{[X]_{\epsilon}}^{n}=\left\{  x\in\mathcal{X}^{n}:\left\vert \frac
1{n}N(a|x^{n})-p_{X}(a)\right\vert \leq\epsilon,~\forall a\in\mathcal{X}%
\mbox{ and }N(a|x^{n})=0\mbox{ whenever } p_{X}(x)=0\right\}  .
\]

The definitions of type and strong typicality can be extended to joint and
conditional distributions in a similar manner \cite{CsiszarKorner}. The
following results concerning typical sets will be used in the sequel. For any
$\epsilon>0$, we have
\begin{equation}
\label{type_d1}\left\vert \frac1{n}\log|T_{[X]_{\epsilon}}^{n}%
|-H(X)\right\vert \leq\epsilon
\end{equation}
and
\begin{equation}
\label{type_d2}\Pr(X^{n}\in T_{[X]_{\epsilon}}^{n})\geq1-\epsilon
\end{equation}
for sufficiently large $n$. Given a joint distribution $p_{XY}$, if the i.i.d.
sequences $(x^{n},y^{n})\sim p_{X}^{n}p_{Y}^{n}$, where $P_{X}^{n}$ and
$P_{Y}^{n}$ are $n$-fold products of the marginals $p_{X}$ and $p_{Y}$, then
\[
\mathrm{Pr}\{(x^{n},y^{n})\in T_{[XY]_{\epsilon}}^{n}\}\leq
2^{-n(I(X;Y)-3\epsilon)}.
\]


\subsection{Converse\label{converse prop 1}}

Starting from the Fano inequality, imposing the condition $P_{e}%
^{n}\rightarrow0$ as $n\rightarrow\infty$, we have%
\begin{equation}
H(W_{1},W_{2},W_{3}|Y^{n})\leq n\delta_{n}\label{Fano1}%
\end{equation}
with $\delta_{n}\rightarrow0$ as $n\rightarrow\infty$. Then we also have
$H(W_{1},W_{3}|Y^{n},W_{2})\leq n\delta_{n}$. We can obtain
\begin{align}
n(R_{1}+R_{3})  &  =H(W_{1},W_{3})=H(W_{1},W_{3}|W_{2})\leq I(W_{1}%
,W_{3};Y^{n}|W_{2})+n\delta_{n}\\
&  =\sum_{i=1}^{n}H(Y_{i}|W_{2},Y^{i-1})-H(Y_{i}|W_{1},W_{2},W_{3}%
,Y^{i-1})+n\delta_{n}\\
&  =\sum_{i=1}^{n}H(Y_{i}|X_{2i},W_{2},Y^{i-1})-H(Y_{i}|X_{1i},X_{2i}%
,X_{3i},W_{1},W_{2},W_{3},Y^{i-1})+n\delta_{n}\label{1_3}\\
&  =\sum_{i=1}^{n}H(Y_{i}|X_{2i},U_{2i},Y^{i-1})-H(Y_{i}|X_{1i},X_{2i}%
,X_{3i},U_{1i},U_{2i})+n\delta_{n}\label{1_4}\\
&  \leq\sum_{i=1}^{n}H(Y_{i}|X_{2i},U_{2i})-H(Y_{i}|X_{1i},X_{2i}%
,X_{3i},U_{1i},U_{2i})+n\delta_{n}\\
&  =\sum_{i=1}^{n}I(X_{1i},U_{1i},X_{3i};Y_{i}|X_{2i},U_{2i})+n\delta_{n}\\
&  =\sum_{i=1}^{n}I(X_{1i},X_{3i};Y_{i}|X_{2i},U_{2i})+n\delta_{n}%
\end{align}
where in (\ref{1_3}) we have used the fact that the codewords are function of
the messages, in (\ref{1_4}) we have defined $U_{1i}\triangleq W_{1}$ and
$U_{2i}\triangleq W_{2}$ and used the fact that $Y^{i-1}-$($X_{1i}%
,X_{2i},X_{3i})-Y_{i}$\ forms a Markov chain, and the last equality follows
from the Markov chain relationship $(U_{1i},U_{2i})-(X_{1i},X_{2i},X_{3i})-Y_{i}.$

We can similarly obtain%
\[
n(R_{2}+R_{3})\leq\sum_{i=1}^{n}I(X_{2i},X_{3i};Y_{i}|X_{1i},U_{1i}%
)+n\delta_{n}%
\]
starting from $n(R_{2}+R_{3})\leq I(W_{2},W_{3};Y^{n}|W_{1})+n\delta_{n}$
(which follows from the Fano inequality (\ref{Fano1}) since it implies
$H(W_{2},W_{3}|Y^{n},W_{1})\leq n\delta_{n}$) and
\[
nR_{3}\leq\sum_{i=1}^{n}I(X_{3i};Y_{i}|X_{1i},X_{2i},U_{1i},U_{2i}%
)+n\delta_{n}%
\]
from the inequality $nR_{3}\leq I(W_{3};Y^{n}|W_{1},W_{2})+n\delta_{n}$ (which
follows from (\ref{Fano1}) as $H(W_{3}|Y^{n},W_{1},W_{2})\leq n\delta_{n}).$
From (\ref{Fano1}), we also have:%
\begin{align*}
n(R_{1}+R_{2}+R_{3})  &  \leq I(W_{1},W_{2},W_{3};Y^{n})+n\delta_{n}\\
&  =\sum_{i=1}^{n}H(Y_{i}|Y^{i-1})-H(Y_{i}|W_{1},W_{2},W_{3},Y^{i-1}%
)+n\delta_{n}\\
&  =\sum_{i=1}^{n}H(Y_{i}|Y^{i-1})-H(Y_{i}|X_{1i},X_{2i},X_{3i},W_{1}%
,W_{2},W_{3},Y^{i-1})+n\delta_{n}
\end{align*}
\begin{align*}
&  \leq\sum_{i=1}^{n}H(Y_{i})-H(Y_{i}|X_{1i},X_{2i},X_{3i})+n\delta_{n}=\\
&  =\sum_{i=1}^{n}I(X_{1i},X_{2i},X_{3i};Y_{i})+n\delta_{n}.
\end{align*}

Now, introducing the time-sharing random variable $Q$ independent from
everything else and uniformly distributed over $\{1,..,n\}$ and
defining $X_{j}\triangleq X_{jQ}$ for $j=1,2,3,$ $Y\triangleq Y_{Q} $ and
$U_{j}\triangleq U_{jQ}$ for $j=1,2,$ we get (\ref{region cognitive}). Notice
that the joint distribution satisfies (\ref{pmf}).

\subsection{Achievability}

\textit{Code Construction}: Generate an i.i.d. sequence $Q^{n}$ with marginal $p(q)$
for $i=1,2,...,n.$ Fix a realization of such a sequence $Q^{n}=q^{n}$. Generate $2^{nR_{j}}$
codewords $u_{j}^{n}(w_{j}),$ $w_{j}=1,2,...,2^{nR_{j}}$ also i.i.d. with probability distribution $\prod_{i=1}^{n}%
p(u_{ji}|q_{i}),$ for $j=1,2$. For each pair $w_{1},w_{2},$ generate
$2^{nR_{3}}$ codewords i.i.d. according to $\prod_{i=1}^{n}p(x_{3i}%
|u_{1i}(w_{1}),u_{2i}(w_{2}),q_{i})$, and label these codewords as $x_{3}%
^{n}(w_{1},w_{2},w_{3})$ for $w_{3}\in\lbrack1,2^{nR_{3}}]$. Also generate
$2^{nR_{j}}$ codewords $x_{j}^{n}(w_{j}),$ $j=1,2$, i.i.d. with probability distribution
$\prod_{i=1}^{n}p(x_{ji}|u_{ji}(w_{j}),q_{i}) $ and label them as $x_{j}%
^{n}(w_{j})$ for $w_{j}\in\lbrack1,2^{nR_{j}}]$.

\textit{Encoders}: Given ($w_{1},w_{2},w_{3}$), encoder $j$ transmits
$x_{j}^{n}(w_{j})$, $j=1,2,$ and encoder 3 transmits $x_{3}^{n}(w_{1}%
,w_{2},w_{3}).$

\textit{Decoders}: The decoder looks for a triplet $(\tilde{w}_{1},\tilde
{w}_{2},\tilde{w}_{3})$ such that%
\[
(q^{n},u_{1}^{n}(\tilde{w}_{1}),u_{2}^{n}(\tilde{w}_{2}),x_{1}^{n}(\tilde
{w}_{1}),x_{2}^{n}(\tilde{w}_{2}),x_{3}^{n}(\tilde{w}_{1},\tilde{w}_{2}%
,\tilde{w}_{3}),y^{n})\in T_{[QU_{1}U_{2}X_{1}X_{2}X_{3}Y]_{\epsilon}}^{n}.
\]
If none or more than one such triplet is found, an error is declared.

\textit{Error analysis: } Assume $(w_{1},w_{2},w_{3})=(1,1,1)$ was sent. We
have an error if either the correct codewords are not typical with the
received sequence or there is an incorrect triplet of messages whose
corresponding codewords are typical with the received sequence. Define the
event (conditioned on the transmission of $(w_{1},w_{2},w_{3})=(1,1,1)$)
\[
E_{k,l,m}\triangleq\{(Q^{n},U_{1}^{n}(k),X_{1}^{n}(k),U_{2}^{n}(l),X_{2}%
^{n}(l),X_{3}^{n}(k,l,m),Y^{n})\in T_{[QU_{1}U_{2}X_{1}X_{2}X_{3}Y]_{\epsilon
}}^{n}\}.
\]

From the union bound, the probability of error, averaged over the
random codebooks, is found as
\begin{align*}
P_{e}^{n}= &  \Pr(E_{1,1,1}^{c})\bigcup\cup_{(k,l,m)\neq(1,1,1)}E_{k,l,m},\\
\leq &  \Pr(E_{1,1,1}^{c})+\sum_{(k,m)\neq(1,1),l=1}\Pr(E_{k,1,m}%
)+\sum_{(l,m)\neq(1,1),k=1}\Pr(E_{1,l,m})+\sum_{k\neq1,l\neq1,m\neq1}%
\Pr(E_{k,l,m}).
\end{align*}

From (\ref{type_d2}), $\Pr(E_{1,1,1}^{c})\rightarrow0$ as
$n\rightarrow\infty$. We can also show that for $(k,m)\neq(1,1),l=1$,
\begin{align*}
\Pr(E_{k,1,m})= &  \Pr((q^{n},u_{1}^{n}(k),u_{2}^{n}(1),x_{1}^{n}(k),x_{2}%
^{n}(1),x_{3}^{n}(k,1,m),y^{n})\in T_{[QU_{1}U_{2}X_{1}X_{2}X_{3}Y]_{\epsilon
}}^{n})\\
\leq &  2^{-n(I(X_{1},U_{1},X_{3};Y|X_{2},U_{2},Q)-3\epsilon)}\\
= &  2^{-n(I(X_{1},X_{3};Y|X_{2},U_{2},Q)-3\epsilon)}.
\end{align*}

Similarly, for $(l,m)\neq(1,1)$ and $k=1$, we have
\begin{align*}
P(E_{k,1,m})= &  \Pr((q^{n},u_{1}^{n}(1),u_{2}^{n}(l),x_{1}^{n}(1),x_{2}%
^{n}(l),x_{3}^{n}(1,l,m),y^{n})\in T_{[QU_{1}U_{2}X_{1}X_{2}X_{3}Y]_{\epsilon
}}^{n})\\
\leq &  2^{-n(I(X_{2},X_{3};Y|X_{1},U_{1},Q)-3\epsilon)}.
\end{align*}

The third error event occurs for $k\neq1,l\neq1,m\neq1$, and we have if%
\begin{align*}
P(E_{k,l,m})= &  \Pr((q^{n},u_{1}^{n}(k),u_{2}^{n}(l),x_{1}^{n}(k),x_{2}%
^{n}(l),x_{3}^{n}(k,l,m),y^{n})\in T_{[QU_{1}U_{2}X_{1}X_{2}X_{3}Y]_{\epsilon
}}^{n})\\
\leq &  2^{-n(I(X_{1},X_{2},X_{3};Y|Q)-4\epsilon)}.
\end{align*}

Then, it follows that
\begin{align*}
P_{e}^{n}\leq &  \Pr(E_{1,1,1}^{c})+2^{n(R_{1}+R_{3})}2^{-n(I(X_{1}%
,X_{3};Y|X_{2},U_{2},Q)-3\epsilon)}+2^{n(R_{2}+R_{3})}2^{-n(I(X_{2}%
,X_{3};Y|X_{1},U_{1},Q)-3\epsilon)}\\
&  +2^{n(R_{1}+R_{2}+R_{3})}2^{-n(I(X_{1},X_{2},X_{3};Y|,Q)-4\epsilon)}.
\end{align*}
Letting $\epsilon\rightarrow0$ and $n\rightarrow\infty$, we have a
vanishing error probability given that the inequalities in (\ref{region cognitive})
are satisfied.

\section{Proof of Proposition \ref{prop:2}}

\label{app:2}

\subsection{Converse}

Similar to the converse in Appendix \ref{app:1}, we use the Fano inequality given
in (\ref{Fano1}). Then we have
\begin{align}
n(R_{1}+R_{3}) &  =H(W_{1},W_{3})=H(W_{1},W_{3}|W_{2})\leq I(W_{1},W_{3}%
;Y^{n}|W_{2})+n\delta_{n}\\
&  =\sum_{i=1}^{n}H(Y_{i}|W_{2},Y^{i-1})-H(Y_{i}|W_{1},W_{2},W_{3}%
,Y^{i-1})+n\delta_{n}\\
&  =\sum_{i=1}^{n}H(Y_{i}|X_{2i},W_{2},Y^{i-1})-H(Y_{i}|X_{1i},X_{2i}%
,X_{3i},W_{1},W_{2},W_{3},Y^{i-1})+n\delta_{n}
\end{align}
\begin{align}
&  \leq\sum_{i=1}^{n}H(Y_{i}|X_{2i})-H(Y_{i}|X_{1i},X_{2i},X_{3i})+n\delta
_{n}\label{1_4b}\\
&  =\sum_{i=1}^{n}I(X_{1i},X_{3i};Y_{i}|X_{2i})+n\delta_{n}%
\end{align}
where (\ref{1_4b}) follows from the fact that conditioning reduces entropy and
$(W_{1},W_{2},W_{3},Y^{i-1})-(X_{1i},X_{2i},X_{3i})-Y_{i}$ forms a Markov
chain. The other inequalities follow similarly.

\subsection{Achievability}

\textit{Code Construction}: Generate $2^{nR_{1}}$ codewords $x_{1}^{n}%
(w_{1}),$ $w_{1}\in\lbrack1,2^{nR_{1}}]$ by choosing each $i$-th letter i.i.d.
from probability distribution $p(x_{1}),$ $i=1,2,...,n.$ For each $w_{1},$ generate
$2^{nR_{3}}$ codewords $x_{3}^{n}(w_{1},w_{3}),$ $w_{3}\in\lbrack1,2^{nR_{3}%
}],$ i.i.d. according to $\prod_{i=1}^{n}p(x_{3i}|x_{1i}(w_{1})).$ Finally,
generate $2^{nR_{2}}$ codewords $x_{2}^{n}(w_{2})$ i.i.d. with each letter
drawn according to $p(x_{2}).$

\textit{Encoding} and\textit{\ error analysis} are rather standard (similar
to Appendix A) and are thus omitted.

\section{Proof of Proposition \ref{prop:conf}\label{app:conf}}

\subsection{Converse}

The converse follows from standard arguments based on the Fano inequality
(see, e.g., Appendix \ref{app:1}). Here, for illustration, we
derive only the first bound in (\ref{conf}), i.e., $R_{1}\leq I(X_{1};Y|X_{2}%
,X_{3},U_{1},U_{2})+C_{1},$ as follows. Define as $V_{1}\in\mathcal{V}_{1}$
and $V_{2}\in\mathcal{V}_{2}$ the messages of cardinality $|\mathcal{V}%
_{i}|\leq C_{i}$ sent over the two links from the sources to the relay. Notice
that $V_{i}$ is a function only of $W_{i}$ and that $X_{3}^{n}$ is a function
only of $W_{3},V_{1}$ and $V_{2}.$ Considering decoding of $W_{1},$ from the
Fano inequality%
\[
H(W_{1}|Y^{n},V_{1},V_{2},W_{2})\leq n\delta_{n},
\]
we get%
\begin{align*}
nR_{1} &  =H(W_{1}|W_{2})\leq I(W_{1};Y^{n},V_{1},V_{2}|W_{2})+n\delta_{n}\\
&  \leq I(W_{1};V_{1}|W_{2})+I(W_{1};V_{2}|V_{1},W_{2})+I(W_{1};Y^{n}%
|W_{2},V_{1},V_{2})+n\delta_{n}\\
&  \leq nC_{1}+I(W_{1};Y^{n}|W_{2},V_{1},V_{2})+n\delta_{n}\\
&  =nC_{1}+\sum_{i=1}^{n}H(Y_{i}|Y^{i-1},W_{2},V_{1},V_{2})-H(Y_{i}%
|Y^{i-1},W_{1},W_{2},V_{1},V_{2})
\end{align*}
\begin{align*}
&  =nC_{1}+\sum_{i=1}^{n}H(Y_{i}|Y^{i-1},X_{2i},X_{3i},W_{2},U_{1i}%
,U_{2i})-H(Y_{i}|Y^{i-1},X_{1i},X_{2i},X_{3i},W_{1},W_{2},U_{1i},U_{2i})\\
&  \leq nC_{1}+\sum_{i=1}^{n}H(Y_{i}|X_{2i},X_{3i},U_{1i},U_{2i}%
)-H(Y_{i}|X_{1i},X_{2i},X_{3i},U_{1i},U_{2i})\\
&  =nC_{1}+\sum_{i=1}^{n}I(X_{1i};Y_{i}|X_{2i},X_{3i},U_{1i},U_{2i}),
\end{align*}
where in the third line we have used the facts that $I(W_{1};V_{1}|W_{2})\leq
H(V_{1})\leq nC_{1}$ and $I(W_{1};V_{2}|V_{1},W_{2})=0$ and the definitions
$U_{1i}=V_{1}$ and $U_{2i}=V_{2}.$ The proof is concluded similarly to
Appendix \ref{app:1}.

\subsection{Achievability}

\textit{Code Construction}: Split the message of the terminals as
$W_{j}=[W_{jp}$ $W_{jc}]$ with $j=1,2,$ where $W_{jp}$ stands for the
``private'' message sent by each terminal without the help of the relay and
$W_{jc}$ for the ``common'' message conveyed to the destination with the help
of the relay. The corresponding rates are $R_{1}=R_{1p}+R_{1c}$ and
$R_{2}=R_{2p}+R_{2c}.$ Generate a sequence $Q^{n}$ i.i.d. using $p(q)$ for
$i=1,2,...,n.$ Fix a realization of such a sequence $Q^{n}=q^{n}$. Generate $2^{nR_{jc}}$
codewords $u_{j}^{n}(w_{jc}),$ $w_{jc}\in\lbrack1,2^{nR_{jc}}]$ by choosing
each $i$th letter independently with probability $p(u_{j}|q_{i}),$
$i=1,2,...,n,$ for $j=1,2.$ For each $w_{jc},$ generate $2^{nR_{jp}}$
codewords $x_{j}^{n}(w_{jc},w_{jp}),$ $j=1,2,$ $w_{jp}\in\lbrack1,2^{nR_{jp}%
}],$ i.i.d. with each letter drawn according to $p(x_{j}|u_{ji}(w_{jc}%
),q_{i}).$ Finally, for each pair $w_{1c},w_{2c},$ generate $2^{nR_{3}}$
codewords $x_{3}^{n}(w_{1c},w_{2c},w_{3}),$ $w_{3}\in\lbrack1,2^{nR_{3}}],$
i.i.d. according to $p(x_{3}|u_{1i}(w_{1c}),u_{2i}(w_{2c}),q_{i}).$

\textit{Encoders}: Given the messages and the arbitrary rate splits at the
transmitters ($w_{1}=[w_{1p}~w_{1c}],w_{2}=[w_{2p}~w_{2c}],w_{3}$), encoder 1
and encoder 2 send the messages $w_{1c}$ and $w_{2c},$ respectively, over the
finite-capacity channels which are then known at the relay before
transmission. Terminal 1 and terminal 2 then transmit $x_{j}^{n}(w_{jc}%
,w_{jp})$, $j=1,2,$ and the relay transmits $x_{3}^{n}(w_{1c},w_{2c},w_{3})$.

The rest of the proof follows similarly to Appendix \ref{app:1} by exploiting
the results in Sec. VII of \cite{Slepian}.

\section{Proof of Proposition \ref{p:achCF}}

\label{app:achCF}

We use the classical block Markov encoding for achievability, and we assume
$|\mathcal{Q}|=1$ for the sake of brevity of the presentation. Generalization to
arbitrary finite cardinalities follows from the usual techniques (see, e.g.,
Appendix \ref{app:1}).

\textit{Codebook generation: } Generate $2^{nR_{k}}$ i.i.d. codewords
$x_{k}^{n}$ from probability distribution $p(x_{k}^{n})=\prod_{i=1}^{n}p(x_{ki})$ for
$k=1,2$. Label each codeword, for $k=1,2$, as $x_{k}^{n}(w_{k})$, where
$w_{k}\in\lbrack1,2^{nR_{k}}]$. Generate $2^{n(R_{0}+R_{3})}$ i.i.d. codewords
$x_{3}^{n}$ from probability distribution $p(x_{3}^{n})=\prod_{i=1}^{n}p(x_{3i}) $. Label
each codeword as $x_{3}^{n}(s,w_{3})$, where $s\in\lbrack1,2^{nR_{0}}]$ and
$w_{3}\in\lbrack1,2^{nR_{3}}]$. Also, for each $x_{3}^{n}(s,w_{3})$, generate
$2^{nR_{0}}$ i.i.d. sequences $\hat{y}_{3}^{n}$ from probability distribution $p(\hat
{y}_{3}^{n}|x_{3}^{n}(s,w_{3}))=\prod_{i=1}^{n}p(\hat{y}_{3i}|x_{3}%
^{n}(s,w_{3}))$, where we define
\[
p(\hat{y}_{3}|x_{3})=\sum_{x_{1},x_{2},y_{1},y_{2},y_{3}}p(x_{1}%
)p(x_{2})p(y_{1},y_{2},y_{3}|x_{1},x_{2},x_{3})p(\hat{y}_{3}|y_{3},x_{3}).
\]
We label these sequences as $\hat{y}_{3}(m|s,w_{3})$, where
$m\in\lbrack1,2^{nR_{0}}]$, $s\in\lbrack1,2^{nR_{0}}]$ and $w_{3}\in
\lbrack1,2^{nR_{3}}].$

\textit{Encoding: } Let $(w_{1i}, w_{2i}, w_{3i})$ be the message to be
transmitted in block $i$, and assume that $( \hat{Y}_{3}^{n}(s_{i}|s_{i-1},
w_{3,i-1}), Y_{3}^{n}(i-1), x_{3}^{n}(s_{i-1}, w_{3,i-1}))$ are jointly
typical. Then the codewords $x_{1}^{n}(w_{1i})$, $x_{2}^{n}(w_{2i})$ and
$x_{3}^{n}(s_{i}, w_{3i})$ will be transmitted in block $i$.

\textit{Decoding: } After receiving $\hat{y}_{3}^{n}(i)$, the relay finds the
index $s_{i+1}$ such that
\[
(\hat{y}_{3}^{n}(s_{i+1}|s_{i},w_{3,i}),y_{3}^{n}(i),x_{3}^{n}(s_{i}%
,w_{3i}))\in T_{[\hat{Y}_{3}Y_{3}X_{3}]_{\epsilon}}^{n}.
\]
For large enough $n$, there will be such $s_{i+1}$ with high probability if
\[
R_{0}>I(Y_{3};\hat{Y}_{3}|X_{3}).
\]
We fix $R_{0}=I(Y_{3};\hat{Y}_{3}|X_{3})+\epsilon$.

At the end of block $i$, the receiver $k$ finds indices $\hat{s}_{i}^{k}$ and
$\hat{w}_{3i}^{k}$ such that
\[
(x_{3}^{n}(\hat{s}_{i}^{k},\hat{w}_{3i}^{k}),y_{k}^{n}(i))\in T_{[X_{3}%
Y_{k}]_{\epsilon}}^{n}%
\]
and
\[
(\hat{y}_{3}^{n}(\hat{s}_{i}^{k}|s_{i-1}^{k},w_{3,i-1}^{k}),x_{3}^{n}%
(s_{i-1}^{k},w_{3,i-1}^{k}),y_{k}^{n}(i-1))\in T_{[\hat{Y}_{3}X_{2}%
Y_{k}]_{\epsilon}}^{n}%
\]
are simultaneously satisfied, assuming that $s_{i-1}^{k}$ and $w_{3,i-1}^{k}$
have been previously correctly estimated. Receiver $k$ will find the correct
pair $(s_{i},w_{3i})$ with high probability provided that $n$ is large enough
and that
\[
R_{3}+I(Y_{3};\hat{Y}_{3}|X_{3},Y_{k})<I(X_{3};Y_{k}).
\]
Assuming that this condition is satisfied so that $\hat{s}_{i}^{k}=s_{i}^{k}$
and $\hat{w}_{3,i-1}^{k}=w_{3,i-1}^{k}$; using both $\hat{y}_{3}^{n}(s_{i}%
^{k}|s_{i-1}^{k},w_{3,i-1}^{k})$ and $y_{k}^{n}(i)$, the receiver $k$ then
declares that $(\hat{w}_{1,i-1}^{k},\hat{w}_{2,i-1}^{k})$ was sent in block
$i-1$ if
\[
(x_{1}^{n}(\hat{w}_{1,i-1}^{k}),x_{2}^{n}(\hat{w}_{2,i-1}^{k}),\hat{y}_{3}%
^{n}(s_{i}^{k}|s_{i-1}^{k},w_{3,i-1}^{k}),y_{k}^{n}(i-1),x_{3}^{n}(s_{i-1}%
^{k},w_{3,i-1}^{k}))\in T_{[X_{1}X_{2}X_{3}\hat{Y}_{3}Y_{k}]_{\epsilon}}^{n}.
\]
We have $(\hat{w}_{1,i-1}^{k},\hat{w}_{2,i-1}^{k})=(w_{1,i-1},w_{2,i-1})$ with
high probability provided that $n$ is large enough,
\begin{align}
R_{1} &  <I(X_{1};Y_{k},\hat{Y}_{3}|X_{3},X_{2}),\nonumber\\
R_{2} &  <I(X_{2};Y_{k},\hat{Y}_{3}|X_{3},X_{1}),
\end{align}
and
\[
R_{1}+R_{2}<I(X_{1},X_{2};Y_{k},\hat{Y}_{3}|X_{3}),
\]
for $k=1,2$.

\section{Proof of Proposition \ref{p:outerbound}}

\label{app:outbnd}

To simplify the presentation, here we prove (\ref{outer bound}) for $R_{3}=0.$
The case with $R_{3}>0$ follows similarly by following the same reasoning as
in Appendix \ref{converse prop 1}. From the Fano inequality, for $i=1,2$,
\[
H(W_{1},W_{2}|Y_{i}^{n})\leq n\delta_{n},
\]
where $\delta_{n}\rightarrow0$ for $n\rightarrow\infty$, from which it follows
that $H(W_{1}|Y_{i}^{n})\leq n\delta_{n}$ and $H(W_{2}|Y_{i}^{n})\leq
n\delta_{n}.$ For the relay terminal, we have
\begin{align}
H(W_{1}|W_{2},Y_{3}^{n})  &  =H(W_{1}|W_{2},X_{2}^{n},Y_{3}^{n},X_{3}%
^{n})\label{fano_s1}\\
&  \leq H(W_{1}|X_{2}^{n},X_{3}^{n})\label{fano_s1b}\\
&  =H(W_{1}|X_{2}^{n},X_{3}^{n},Y_{2}^{n})\label{fano_s2}\\
&  \leq H(W_{1}|Y_{2}^{n})\label{fano_s3}\\
&  \leq n\delta_{n},\label{fano_s4}%
\end{align}
where (\ref{fano_s1}) follows since $X_{2}^{n}$ is a function of $W_{2}$ and
$X_{3}^{n}$ is a function of $Y_{3}^{n}$; (\ref{fano_s2}) follows since
$Y_{2}^{n}-(X_{2}^{n},X_{3}^{n})-W_{1}$ form a Markov chain based on the
assumption in (\ref{ch_ass_1}); (\ref{fano_s3}) follows since conditioning
reduces entropy; and finally (\ref{fano_s4}) follows from the Fano inequality.
Similarly, we can also show
\[
H(W_{2}|W_{1},Y_{3}^{n})\leq n\delta_{n}.
\]

We also define the auxiliary random variables $U_{1i}\triangleq W_{1}$ and
$U_{2i}\triangleq W_{2}$, for $i=1,\ldots,n$. It follows that:
\begin{align}
nR_{1}  &  =H(W_{1})=H(W_{1}|W_{2})\label{rate1_inq1}\\
&  \leq I(W_{1};Y_{3}^{n}|W_{2})+n\delta_{n}\label{rate1_inq2}\\
&  =\sum_{i=1}^{n}I(W_{1};Y_{3i}|W_{2},Y_{3}^{i-1})+n\delta_{n}\nonumber\\
&  =\sum_{i=1}^{n}H(Y_{3i}|Y_{3}^{i-1},W_{2})-H(Y_{3i}|W_{1},W_{2},Y_{3}%
^{i-1})+n\delta_{n}\nonumber\\
&  =\sum_{i=1}^{n}H(Y_{3i}|Y_{3}^{i-1},W_{2},X_{2i},X_{3i})-H(Y_{3i}%
|W_{1},W_{2},Y_{3}^{i-1},X_{1i},X_{2i},X_{3i})+n\delta_{n}\label{rate1_inq3}\\
&  \leq\sum_{i=1}^{n} H(Y_{3i}|W_{2}, X_{2i},X_{3i})-H(Y_{3i}|X_{1i}%
,X_{2i},X_{3i})+n\delta_{n}\label{rate1_inq4}\\
&  \leq\sum_{i=1}^{n} H(Y_{3i}|U_{2i}, X_{2i},X_{3i})-H(Y_{3i}|U_{2i},
X_{1i},X_{2i},X_{3i})+n\delta_{n}\label{rate1_inq4b}\\
&  =\sum_{i=1}^{n}I(X_{1i};Y_{3i}|U_{2i}, X_{2i},X_{3i})+n\delta
_{n},\label{rate1_inq5}%
\end{align}
where (\ref{rate1_inq2}) follows from (\ref{fano_s4}); (\ref{rate1_inq3})
follows as $X_{1i}$ and $X_{2i}$ are functions of $W_{1}$ and $W_{2}$,
respectively, and $X_{3i}$ is a function of $Y_{3}^{i-1}$; (\ref{rate1_inq4})
follows from the fact that conditioning reduces entropy and also the fact that
$Y_{3i}-(X_{1i},X_{2i},X_{3i})-(W_{1},W_{2},Y_{3}^{i-1})$; and again
(\ref{rate1_inq4b}) follows from the fact that conditioning reduces entropy and
the definition of auxiliary random variable $U_{2i}$. Similarly, we can show
that%
\begin{equation}
nR_{2} \leq\sum_{i=1}^{n} I(X_{2i};Y_{3i}|U_{1i},X_{1i},X_{3i}%
).\label{rate1_inq6}%
\end{equation}

Next, we consider the bounds due to decoding at receivers. Focusing on the
first message and the first receiver, we have
\begin{align}
nR_{1}  &  =H(W_{1}|W_{2})\leq I(W_{1};Y_{1}^{n}|W_{2})+n\delta_{n}\nonumber\\
&  =\sum_{i=1}^{n}H(Y_{1i}|W_{2},Y_{1}^{i-1})-H(Y_{1i}|W_{1},W_{2},Y_{1}%
^{i-1})+n\delta_{n}\nonumber\\
&  =\sum_{i=1}^{n}H(Y_{1i}|W_{2},Y_{1}^{i-1})-H(Y_{1i}|W_{1},W_{2}%
,X_{1i},X_{2i},X_{3i},Y_{1}^{i-1})+n\delta_{n}\label{rate1b_inq3}\\
&  \leq\sum_{i=1}^{n}H(Y_{1i}|W_{2})-H(Y_{1i}|X_{1i},X_{2i},X_{3i}%
)+n\delta_{n}\label{rate1b_inq4}
\end{align}
\begin{align}
&  =\sum_{i=1}^{n}H(Y_{1i}|W_{2})-H(Y_{1i}|X_{1i},X_{3i})+n\delta
_{n}\label{rate1b_inq5}\\
&  \leq\sum_{i=1}^{n}H(Y_{1i}|U_{2i})-H(Y_{1i}|X_{1i},X_{3i},U_{2i}%
)+n\delta_{n}\label{rate1b_inq6}\\
&  =\sum_{i=1}^{n}I(X_{1i},X_{3i};Y_{1i}|U_{2i})+n\delta_{n}%
\label{rate1b_inq7}%
\end{align}
where (\ref{rate1b_inq3}) follows since $X_{1i}$ and $X_{2i}$ are functions of
$W_{1}$ and $W_{2}$, respectively, and $X_{3i}$ is a function of $Y_{3}^{i-1}%
$; (\ref{rate1b_inq4}) follows as conditioning reduces entropy and
$Y_{1i}-(X_{1i},X_{2i},X_{3i})-(W_{1},W_{2},Y_{1}^{i-1})$ forms a Markov
chain; (\ref{rate1b_inq5}) follows since $Y_{1i}-(X_{1i},X_{3i})-X_{2i}$ forms
a Markov chain; and finally in (\ref{rate1b_inq6}) we simply used the
definition of $U_{2i}$ and the fact that conditioning reduces entropy. We can
similarly obtain
\[
nR_{2}\leq\sum_{i=1}^{n}I(X_{2i},X_{3i};Y_{2i}|U_{1i})+n\delta_{n}.
\]

We now focus on the first message and the second receiver. We have
\begin{align}
nR_{1}  & = H(W_{1})\nonumber\\
&  \leq I(W_{1};Y_{2}^{n}|W_{2})+n\delta_{n}\nonumber\\
&  =\sum_{i=1}^{n}H(Y_{2i}|W_{2},Y_{2}^{i-1})-H(Y_{2i}|W_{1},W_{2},Y_{2}%
^{i-1})+n\delta_{n}\nonumber\\
&  =\sum_{i=1}^{n}H(Y_{2i}|W_{2},Y_{2}^{i-1},X_{2i})-H(Y_{2i}|W_{1}%
,W_{2},X_{1i},X_{2i},X_{3i},Y_{2}^{i-1})+n\delta_{n}\label{rate1c_inq3}\\
&  \leq\sum_{i=1}^{n}H(Y_{1i}|W_{2},X_{2i})-H(Y_{2i}|X_{1i},X_{2i}%
,X_{3i})+n\delta_{n}\label{rate1c_inq4}\\
&  =\sum_{i=1}^{n}H(Y_{2i}|W_{2},X_{2i})-H(Y_{2i}|X_{2i},X_{3i})+n\delta
_{n}\label{rate1c_inq5}\\
&  \leq\sum_{i=1}^{n}H(Y_{2i}|U_{2i},X_{2i})-H(Y_{2i}|X_{2i},X_{3i}%
,U_{2i})+n\delta_{n}\label{rate1c_inq6}\\
&  =\sum_{i=1}^{n}I(X_{3i};Y_{2i}|U_{2i},X_{2i})+n\delta_{n}%
.\label{rate1c_inq7}%
\end{align}
We can similarly obtain
\[
nR_{2}\leq\sum_{i=1}^{n}I(X_{3i};Y_{1i}|U_{1i},X_{1i})+n\delta_{n}.
\]
We also have
\begin{align*}
n(R_{1}+R_{2})  &  \leq I(W_{1},W_{2};Y_{1}^{n})+n\delta_{n}\\
&  =\sum_{i=1}^{n}H(Y_{1i}|Y_{1}^{i-1})-H(Y_{1i}|W_{1},W_{2},Y_{1}%
^{i-1})+n\delta_{n}\\
&  =\sum_{i=1}^{n}H(Y_{1i}|Y_{1}^{i-1})-H(Y_{1i}|X_{1i},X_{2i},X_{3i}%
,W_{1},W_{2},Y_{1}^{i-1})+n\delta_{n}\\
&  \leq\sum_{i=1}^{n}H(Y_{1i})-H(Y_{1i}|X_{1i},X_{2i},X_{3i})+n\delta_{n}\\
&  \leq\sum_{i=1}^{n}H(Y_{1i})-H(Y_{1i}|X_{1i},X_{3i})+n\delta_{n}
\end{align*}
\begin{align*}
&  =\sum_{i=1}^{n}I(X_{1i},X_{3i};Y_{1i})+n\delta_{n},
\end{align*}
and similarly
\[
n(R_{1}+R_{2})\leq\sum_{i=1}^{n}I(X_{2i},X_{3i};Y_{2i})+n\delta_{n}.
\]
Now, introducing the time-sharing random variable $Q$ uniformly distributed in
the set $\{1,2,..,n\}$ and defining $X_{j}=X_{jQ}$ for $j=1,2,3,$
$Y_{j}=Y_{jQ} $ and $U_{j}=U_{jQ}$ for $j=1,2,$ we get (\ref{outer bound})
(for $R_{3}=0)$.

\section{Proof of Proposition \ref{p:binary}}

\label{app:binary}

We first prove the converse showing that (\ref{binary example}) serves as an
outer bound, and prove the direct part describing a structured coding scheme
that achieves the outer bound.

To prove the converse, it is sufficient to consider the outer bound given by
(\ref{outer bound}) as applied to the channel characterized by (\ref{bsc}),
and show that an input distribution (\ref{distribution outer}) with
$X_{1},X_{2},X_{3},U_{1},U_{2}\sim\mathcal{B}(1/2)$ and independent of each
other maximizes all the mutual information terms. To this end, notice that in
the outer bound (\ref{outer bound}) with $R_{3}=0$ ignoring all the
constraints involving auxiliary random variables can only enlarge the region,
so that we have the conditions:
\begin{align}
R_{1} &  \leq I(X_{1};Y_{3}|X_{2},X_{3},Q),\\
R_{2} &  \leq I(X_{2};Y_{3}|X_{1},X_{3},Q)
\end{align}
and
\begin{align}
R_{1}+R_{2} &  \leq\min\{I(X_{1},X_{3};Y_{1}|Q),I(X_{2},X_{3};Y_{2}|Q)\}.
\end{align}
We can further write
\begin{align*}
I(X_{1};Y_{3}|X_{2},X_{3},Q) &  =H(Y_{3}|X_{2},X_{3},Q)-H(Y_{3}|X_{1}%
,X_{2},X_{3},Q)\\
&  \leq H(Y_{3})-H_{b}(\varepsilon_{3})\leq1-H_{b}(\varepsilon_{3}),
\end{align*}
and
\begin{align*}
I(X_{1},X_{3};Y_{1}|Q) &  =H(Y_{1}|Q)-H(Y_{1}|X_{1},X_{3},Q)\\
&  \leq H(Y_{1})-H_{b}(\varepsilon_{1})\leq1-H_{b}(\varepsilon_{1}).
\end{align*}
We can see that the inequalities hold with equality under the above stated
input distribution, which concludes the proof of the converse.

We now prove the direct part of the proposition. First, consider $R_{1}\geq
R_{2} $. Transmission is organized into $B$ blocks of size $n$ bits. In each
of the first $B-1$ blocks, say the $b$th, the $j$-th transmitter, $j=1,2,$
sends $nR_{j}$ new bits, conventionally organized into a $1\times\left\lfloor
nR_{j}\right\rfloor $ vector $\mathbf{u}_{j,b}$. Moreover, encoding at the
transmitters is done using the same binary linear code characterized by a
$\left\lfloor nR_{1}\right\rfloor \times n$ random binary generator matrix
$\mathbf{G}$ with i.i.d. entries $\mathcal{B}(1/2).$

Specifically, as in \cite{Nam:IZS:08}, terminal 1 transmits $\mathbf{x}%
_{1,b}=\mathbf{u}_{1,b}\mathbf{G}$ and terminal 2 transmits $\mathbf{x}%
_{2,b}=[\mathbf{0}$ $\mathbf{u}_{2,b}]\mathbf{G}$ where the all-zero vector is
of size $1\times\left\lfloor nR_{1}\right\rfloor -\left\lfloor nR_{2}%
\right\rfloor $ (zero-padding). Since capacity-achieving random linear codes
exist for BS channels, we assume that $\mathbf{G}$ is the generating matrix
for such a capacity achieving code.

We define $\mathbf{u}_{3,b}\triangleq\mathbf{u}_{1,b}\oplus\lbrack\mathbf{0}$
$\mathbf{u}_{2,b}]$. The relay can then decode $\mathbf{u}_{3,b}$ from the
received signal $\mathbf{y}_{3,b}=\mathbf{u}_{3,b}\mathbf{G+z}_{3}$ since
$\mathbf{x}_{1,b}\oplus\mathbf{x}_{2,b}$ is also a codeword of the code
generated by $\mathbf{G}$. This occurs with vanishing probability of error if
(\ref{bin example 1}) holds (see, e.g., \cite{medard}). In the following
$(b+1)$-th block, the relay encodes $\mathbf{u}_{3,b}$ using an independent
binary linear code with an $\left\lfloor nR_{1}\right\rfloor \times n$ random
binary generator matrix $\mathbf{G}_{3}$ as $\mathbf{x}_{3,b+1}=\mathbf{u}%
_{3,b}\mathbf{G}_{3}$. We use the convention that the signal sent by the relay
in the first block is $\mathbf{x}_{3,1}=\mathbf{0}$ or any other known sequence.

At the end of the first block ($b=1$), where the relay sends a known signal
(that can be canceled by both receivers), the $j$-th receiver can decode the
current $nR_{j}$ bits $\mathbf{u}_{j,1}$ from the $j$th transmitter if
$R_{j}\leq1-H_{b}(\varepsilon_{j}).$ Under this condition, we can now consider
the second block, or any other $(b+1)$-th block, assuming that the $j$-th
receiver already knows $\mathbf{u}_{j,b}$. In the $(b+1)$-th block, the first
receiver sees the signal $\mathbf{y}_{1,b+1}=\mathbf{u}_{1,b+1}%
\mathbf{G}\oplus\mathbf{u}_{3,b}\mathbf{G}_{3}\oplus\mathbf{z}_{1}.$ However,
since $\mathbf{u}_{1,b}$ is known at the first receiver, it can be canceled
from the received signal, leading to $\mathbf{y}_{1,b+1}^{\prime}%
=\mathbf{u}_{1,b+1}\mathbf{G}\oplus\mathbf{u}_{2,b}\mathbf{G}_{3}^{\prime
}\oplus\mathbf{z}_{1},$ where $\mathbf{G}_{3}^{\prime}$ is a $\left\lfloor
nR_{2}\right\rfloor \times n$ matrix that contains the last $\left\lfloor
nR_{2}\right\rfloor $ rows of $\mathbf{G}_{3}.$ Due to the optimality of
random linear codes over the BS MAC (see, e.g., \cite{medard}), $\mathbf{u}%
_{1,b+1}$ and $\mathbf{u}_{2,b}$ are correctly decoded by the first receiver
if $R_{1}+R_{2}\leq1-H_{b}(\varepsilon_{1})$. Repeating this argument for the
second receiver and then considering the case $R_{1}\geq R_{2}$ concludes the proof.

\section{Proof of Proposition \ref{p:Lattice}}

\label{app:Lattice}

We first give a brief overview of lattice codes (see \cite{Erez:IT:04},
\cite{Wilson:IT:08} for further details). An $n$-dimensional lattice $\Lambda$
is defined as
\[
\Lambda=\{GX:X\in\mathds{Z}^{n}\},
\]
where $G\in\mathbf{R}^{n}$ is the generator matrix. For any $x\in
\mathds{R}^{n}$, the quantization of $X$ maps $X$ to the nearest lattice point
in Euclidean distance:
\[
Q_{\Lambda}(X)\triangleq\arg\min_{Q\in\Lambda}\Vert X-Q\Vert.
\]
The mod operation is defined as
\[
X\mod\Lambda=X-Q_{\Lambda}(X).
\]
The fundamental Voronoi region $\mathcal{V}(\Lambda)$ is defined as
$\mathcal{V}(\Lambda)=\{X:Q_{\Lambda}(X)=0\}$, whose volume is denoted by
$V(\Lambda)$ and is defined as $V(\Lambda)=\int_{\mathcal{V}(\Lambda)}dX$. The
second moment of lattice $\Lambda$ is given by
\[
\sigma^{2}(\Lambda)=\frac{1}{nV(\Lambda)}\int_{\mathcal{V}(\Lambda)}\Vert
X\Vert^{2}dX,
\]
while the normalized second moment is defined as
\[
G(\Lambda)=\frac{\sigma^{2}(\Lambda)}{V(\Lambda)^{1/n}}=\sigma^{2}%
(\Lambda)=\frac{1}{nV(\Lambda)}\int_{\mathcal{V}(\Lambda)}\Vert X\Vert^{2}dX.
\]
We use a nested lattice structure as in \cite{Zamir:IT:02}, where $\Lambda
_{c}$ denotes the coarse lattice and $\Lambda_{f}$ denotes the fine lattice
and we have $\Lambda_{c}\subseteq\Lambda_{f}$. Both transmitters use the same
coarse and fine lattices for coding. We consider lattices such that
$G(\Lambda_{c})\approx\frac{1}{2\pi e}$ and $G(\Lambda_{f})\approx\frac
{1}{2\pi e}$, whose existence is shown in \cite{Zamir:IT:02}. In nested
lattice coding, the codewords are the lattice points of the fine lattice that
are in the fundamental Voronoi region of the coarse lattice. Moreover, we
choose the coarse lattice (i.e., the shaping lattice) such that $\sigma
^{2}(\Lambda_{c})=P$ to satisfy the power constraint. The fine lattice is
chosen to be good for channel coding, i.e., it achieves the Poltyrev exponent
\cite{Zamir:IT:02}.

We use a block Markov coding structure, that is the messages are coded into
$B$ blocks, and are transmitted over $B+1$ channel blocks. The relay forwards
the information relating to the messages from each block over the next channel
block. The relay is kept silent in the first channel block, while the
transmitters are silent in the last one. The receivers decode the messages
from the transmitters and the relay right after each block. Since there is no
coherent combining, transmitters send only new messages at each channel block,
thus sequential decoding with a window size of one is sufficient. We explain
the coding scheme for two consecutive channel blocks dropping the channel
block index in the expressions.

Each transmitter $i$ maps its message $W_{i}$ to a fine lattice point $V_{i}
\in\Lambda_{f}\cap\mathcal{V}(\Lambda_{c})$, $i=1,2$. Each user employs a
dither vector $U_{i}$ which is independent of the dither vectors of the other users and of the messages and
is uniformly distributed over $\mathcal{V}(\Lambda_{c})$. We assume all the
terminals in the network know the dither vectors. Now the transmitted codeword
from transmitter $i$ is given by $X_{i}= (V_{i} - U_{i}) \mod \Lambda_{c}$. It
can be shown that $X_{i}$ is also uniform over $\mathcal{V}(\Lambda_{c})$.

At the end of each block, we want the relay to decode $V \triangleq(V_{1} +
V_{2}) \mod \Lambda_{c}$ instead of decoding both messages. Following
\cite{Wilson:IT:08} (with proper scaling to take care of the channel gain
$\gamma$), it is possible to show that $V$ can be decoded at the relay if
\begin{equation}
\label{cond_dec_relay}R  \leq \frac1{n} \log_{2}|\Lambda_{f}\cap\mathcal{V}%
(\Lambda_{c})| < \frac12\log\left( \frac12+\gamma^{2} P\right) .
\end{equation}


Then in the next channel block, while the transmitters send new information,
the relay terminal broadcasts the index of $V$ to both receivers. The relay
uses rate-splitting \cite{Rimoldi:IT:96}, and transmits each part of the $V$
index using a single-user random code. Let $R_{1}$ and $R_{2}$ be the rates of
the two codes the relay uses, with power allocation $\delta$ and $P-\delta$,
respectively. Each receiver applies successive decoding; the codes from the
relay terminal are decoded using a single-user typicality decoder, while the
signal from the transmitter is decoded by a Euclidean lattice decoder.
Successful decoding is possible if
\begin{align}
R_{1}  & \leq \frac12 \log\left( 1+ \eta^{2} \delta\right) ,\nonumber\\
R  & \leq \frac12 \log\left( 1+ \frac{P}{1+\eta^{2} \delta} \right)
\mbox{ and}\nonumber
\end{align}
\begin{align}
R_{2}  & \leq \frac12 \log\left( 1+ \frac{\eta^{2}(P-\delta)}{1+\eta^{2}
\delta+P} \right) ,\nonumber
\end{align}
where $R_{1}+R_{2} = R$. This is equivalent to having
\[
R \leq\left\{ \frac12 \log\left( 1+ \eta^{2} P \right) , \frac12 \log\left( 1+
P \right) , \frac1{4} \log\left( 1+ (1+\eta^{2})P \right)  \right\} .
\]
Combining this with (\ref{cond_dec_relay}), we obtain the rate constraint
given in the theorem.

\section{Proof of Proposition \ref{prop:6}}

\label{app:prop6}

The proof follows from \cite{Bross}, and here we only briefly sketch the main
arguments. The first step is to notice that the capacity region for the
Gaussian channel (\ref{GMAC}) with power constraints (\ref{power constraints})
is given by the region (\ref{region cognitive}) where the further constraints
$E[X_{j}^{2}]=\sum_{q=1}^{4}p(q)E[X_{j}^{2}|Q=q]\leq P_{j}$ for $j=1,2,3$ are
imposed on the input distribution. One then fixes a given value of $Q=q$ and
powers $P_{j}(q)\triangleq E[X_{j}^{2}|Q=q]$ and shows that, for any given
input distribution (recall (\ref{pmf})) ($U_{1q},U_{2q},X_{1q},X_{2q},X_{3q}%
$)$\sim p(x_{1},u_{1}|q)p(x_{2},u_{2}|q)p(x_{3}|u_{1},u_{2},q)$ satisfying the
power constraints $E[X_{j}^{2}|Q=q]\leq P_{j}(q)$, one can find a set of
jointly Gaussian variables ($U_{1q}^{\mathcal{G}},U_{2q}^{\mathcal{G}}%
,X_{1q}^{\mathcal{G}},X_{2q}^{\mathcal{G}},X_{3q}^{\mathcal{G}}$) such that:
(\textit{i}) the joint distribution can be factorized as in (\ref{pmf});
(\textit{ii}) the power constraints are satisfied; and (\textit{iii}) all the
mutual information terms in (\ref{region cognitive}) are larger than or equal
to the corresponding values obtained with the original distribution.

Notice that, as discussed in \cite{Bross}, the existence of such a tuple of
Gaussian variables does not follow immediately from the conditional maximum
entropy theorem. In fact, variables satisfying given Markov constraints (as in
(\ref{pmf})) might have a covariance matrix for which a joint Gaussian
distribution with the same Markov constraints cannot be constructed. However,
in our case, similarly to \cite{Bross}, jointly Gaussian random variables
satisfying (\textit{i})-(\textit{iii}) can be found as discussed below.

First, for a given tuple $(U_{1q},U_{2q},X_{1q},X_{2q},X_{3q})$, construct
another tuple ($V_{1q},V_{2q},X_{1q},X_{2q},X_{3q}$) with $V_{1q}%
=E[X_{1}|U_{1}]-E[X_{1}]$ and $V_{2q}=E[X_{2}|U_{2}]-E[X_{2}].$ It can be
readily seen that with this new input distribution, all the mutual information
terms in (\ref{region cognitive}) are larger than or equal to the
corresponding values for $(U_{1q},U_{2q},X_{1q},X_{2q},X_{3q})$ (property
(\textit{iii}) above)$.$ In fact, the non-conditional term
(\ref{non-conditional}) is unchanged (due to the fact the joint distribution
of $X_{1q},X_{2q},X_{3q}$ has not changed), while the remaining terms, which
contain conditioning with respect to $V_{1q},V_{2q},$ are increased, since
$V_{1q},V_{2q}$ are deterministic functions of $U_{1q},U_{2q},$ respectively
\cite{Bross}. Now, define $(U_{1q}^{\mathcal{G}},U_{2q}^{\mathcal{G}}%
,X_{1q}^{\mathcal{G}},X_{2q}^{\mathcal{G}},X_{3q}^{\mathcal{G}})$ as the
zero-mean jointly Gaussian tuple with the same covariance matrix as
$(V_{1q},V_{2q},X_{1q},X_{2q},X_{3q})$. For this distribution the power
constraints are clearly satisfied (property (\textit{ii}) above). Moreover, we
can show that this jointly Gaussian distribution factorizes as (\ref{pmf})
(property (\textit{i})). This follows similarly to Lemma 6 in \cite{Bross}. In
fact, $X_{1q},X_{2q}$ are independent and, since $U_{1q},U_{2q}$ are
independent, so are $V_{1q},V_{2q}$; as a consequence, the covariance matrix
of $(U_{1q}^{\mathcal{G}},U_{2q}^{\mathcal{G}},X_{1q}^{\mathcal{G}}%
,X_{2q}^{\mathcal{G}},X_{3q}^{\mathcal{G}})$ is fully defined by the
subcovariance matrices of $(U_{1q}^{\mathcal{G}},X_{1q}^{\mathcal{G}}%
,X_{3q}^{\mathcal{G}}) $ and $(U_{2q}^{\mathcal{G}},X_{2q}^{\mathcal{G}%
},X_{3q}^{\mathcal{G}}).$ Now, since both submatrices satisfy the conditions
of Lemma 5 of \cite{Bross}, we can conclude that the jointly Gaussian vector
$(U_{1q}^{\mathcal{G}},U_{2q}^{\mathcal{G}},X_{1q}^{\mathcal{G}}%
,X_{2q}^{\mathcal{G}},X_{3q}^{\mathcal{G}})$ satisfies the desired Markov
conditions (\ref{pmf}). It finally follows from the conditional maximum
entropy theorem that $(U_{1q}^{\mathcal{G}},U_{2q}^{\mathcal{G}}%
,X_{1q}^{\mathcal{G}},X_{2q}^{\mathcal{G}},X_{3q}^{\mathcal{G}})$ is optimal
for a given $Q=q.$ The final step is to use the concavity of the mutual
informations at hand with respect to the powers $P_{j}(q)$ to see that time
sharing (i.e., the variable $Q$ with non-singleton domain) is not necessary.

From the arguments sketched above, we conclude that a zero-mean jointly
Gaussian distribution is optimal, thus we can write
\begin{equation}
X_{j}=\sqrt{\alpha_{j}P_{j}}U_{j}+\sqrt{\bar{\alpha}_{j}P_{j}}S_{j}%
\label{def Xj}%
\end{equation}
\ for $j=1,2$ and
\begin{equation}
X_{3}=\sqrt{\alpha_{3}^{\prime}P_{3}}U_{1}+\sqrt{\alpha_{3}^{\prime\prime
}P_{3}}U_{2}+\sqrt{\bar{\alpha}_{3}P_{3}}S_{3}\label{def X3}%
\end{equation}
with $0\leq\alpha_{1},\alpha_{2},\alpha_{3}^{\prime},\alpha_{3}^{\prime\prime
}\leq1, $ $\alpha_{3}^{\prime}+\alpha_{3}^{\prime\prime}\leq1,$ $\bar{\alpha
}_{1}=1-\alpha_{1},$ $\bar{\alpha}_{2}=1-\alpha_{2}$ and $\bar{\alpha}%
_{3}=1-\alpha_{3}^{\prime}-\alpha_{3}^{\prime\prime},$ and where $U_{i}$ and
$S_{i} $ are independent zero-mean and unit-variance Gaussian random
variables. The capacity region (\ref{region cognitive}) then reads%
\begin{align}
R_{3}  &  \leq\frac12\log(1+\bar{\alpha}_{3}P_{3}),\\
R_{1}+R_{3}  &  \leq\frac12 \log(1+\bar{\alpha}_{1} P_{1}+\bar{\alpha}_{3}
P_{3}+(\sqrt{\alpha_{1} P_{1}}+\sqrt{\alpha_{3}^{\prime}P_{3}})^{2})),\\
R_{2}+R_{3}  &  \leq\frac12 \log(1+\bar{\alpha}_{2} P_{2}+\bar{\alpha}_{3}
P_{3}+(\sqrt{\alpha_{2} P_{2}}+\sqrt{\alpha_{3}^{\prime\prime}P_{3}})^{2}))
\mbox{ and }\\
R_{1}+R_{2}+R_{3}  &  \leq\frac12 \log(1+P_{1}+P_{2}+P_{3}+2\sqrt{\alpha_{1}
P_{1}\alpha_{3}^{\prime}P_{3}}+2\sqrt{\alpha_{2} P_{2}\alpha_{3}^{\prime
\prime}P_{3}}),
\end{align}
where each term is clearly seen to be maximized by $\alpha_{1}=1$ and
$\alpha_{2}=1.$ This concludes the proof.

\end{document}